%% file: main_revised.tex
\documentclass[lettersize,journal]{IEEEtran}
\usepackage{amsmath,amsfonts,amssymb}
\usepackage{algorithmic}
\usepackage{algorithm}
\usepackage{array}
\usepackage{textcomp}
\usepackage{stfloats}
\usepackage{url}
\usepackage{multirow}
\usepackage{tabularx,booktabs} %Self added 
\usepackage{verbatim}
\usepackage{graphicx}
\usepackage{cite}
\usepackage{makecell}
\usepackage{url}
\usepackage{relsize} %Self added
\usepackage{color}
\usepackage{graphicx}
\usepackage{float}
\usepackage{subfigure}
\usepackage{enumitem}
\usepackage{bbm}
\usepackage{amsmath}
\usepackage{stfloats}
\usepackage{amsmath,amssymb}
\usepackage{amsthm}  % 提供 lemma / proof 等环境

\newtheorem{lemma}{Lemma}
\makeatletter
\renewcommand{\maketag@@@}[1]{\hbox{\m@th\normalsize\normalfont#1}}%
\makeatother
\usepackage{hyperref}  %hyperref still needs to be put at the end!

\begin{document}

\title{Movable-Antenna Index Modulation (MA-IM): System Framework and Performance Analysis} 

\author{Bang Huang {\em Member, IEEE},~ Shunyuan Shang, and
Mohamed-Slim Alouini,~\IEEEmembership{Fellow,~IEEE}
\thanks{The authors are with Computer, Electrical and Mathematical Sciences
and Engineering (CEMSE) Division, Department of Electrical and Computer
Engineering, King Abdullah University of Science and Technology (KAUST), Thuwal 23955-6900, Saudi Arabia. (e-mail: bang.huang@kaust.edu.sa;shunyuan.shang@kaust.edu.sa slim.alouini@kaust.edu.sa) (Corresponding author: Shunyuan Shang). 
}
% \thanks{This work was supported in part by the National Nature Science Foundation of China under Grant 62201347, Shanghai Sailing Program under Grant 22YF1428400, and Research Council of Finland (Grants 348515 (UPRISING) and 369116 (6G Flagship)). (\textit{Corresponding author: Yijie Mao})}  

% \thanks{K. Lin and M.-S. Alouini are with the Division of Computer, Electrical and Mathematical Sciences and Engineering, King Abdullah University of Science and Technology, Saudi Arabia (E-mail: kaiqiang.lin@kaust.edu.sa; slim.alouini@kaust.edu.sa).}

% \thanks{Y. Mao is with the School of Information Science and Technology, ShanghaiTech University, Shanghai, China (E-mail: maoyj@shanghaitech.edu.cn).}

% \thanks{O. L. A. López is with the Centre for Wireless Communications, University of Oulu, Finland (Email: onel.alcarazlopez@oulu.fi).}
}
\maketitle
\begin{abstract}
This paper proposes a movable-antenna-based index modulation
(MA-IM) framework that exploits the spatial mobility of a single
reconfigurable antenna to create additional information-bearing
dimensions for next-generation wireless systems.
By discretizing the continuous movable region into a dense set of
candidate sampling points and selecting representative anchors for
indexing, the proposed framework converts spatial degrees of freedom
into a practical modulation resource.
Building on this framework, we develop a family of anchor-selection
strategies with different levels of channel awareness, including
geometry-based, SNR-based, max--min channel-domain, and joint
constellation-aware designs.
For the resulting MA-IM schemes, joint maximum-likelihood (ML)
detectors are derived, along with a low-complexity two-stage detector,
and unified analytical upper bounds on the average bit error
probability (ABEP) are established based on the joint
index--modulation constellation.
The results reveal that directly indexing all sampling points is
generally unreliable, highlighting the necessity of anchor
optimization.
The performance of MA-IM is shown to depend on key system parameters,
including channel richness, spatial correlation, the number of index
states, and the modulation order.
In particular, increasing the number of index states and increasing
the QAM order affect MA-IM in fundamentally different ways, even under
the same transmission rate.
Among the proposed schemes, the joint constellation-aware anchor
design achieves the best error performance, demonstrating that
optimizing channel-domain separation alone is insufficient and that
effective MA-IM design must account for the geometry of the joint
signal constellation.
Simulation results further show that, with properly designed anchors,
MA-IM can approach or even outperform same-spectral-efficiency QAM
baselines.
\end{abstract}

\section{Introduction}
\IEEEPARstart{I}{n}  recent years, mobile communication systems have been rapidly evolving
toward ultra-high capacity, ultra-low latency, and ultra-high energy 
efficiency \cite{Saad2020AVisionof}. However, traditional multi-antenna technologies such as 
multiple-input multiple-output (MIMO) and massive MIMO \cite{Lu2014AnOverviewof} are increasingly limited by hardware scalability, 
RF-chain cost, power consumption, and physical deployment constraints. 
Moreover, spectrum scarcity makes it difficult to meet future capacity 
demands merely by increasing modulation order or bandwidth. These trends 
call for new architectures that can enhance spectral and energy 
efficiency without additional hardware or power overhead \cite{Jaradat2025ASurveyof,Pereira2022AnOverview,LiMin2025Sparse}.

To address the aforementioned challenges, a new class of reconfigurable
antenna technologies has attracted significant interest. Reconfigurable
intelligent surfaces (RIS) 
\cite{LiuLiu2021ReconfigurableIntelligentSurfaces}
enable programmable manipulation of the propagation environment through
large arrays of nearly passive elements, whereas reconfigurable array
antennas achieve flexible beamforming via switching networks or tunable
components \cite{liu2025reconfigurable}. In parallel, fluid antennas (FA) 
\cite{WongShojaeifard2021FluidAntenna,Wong2020PerformanceLimits,WongTong2022FluidAntenna}
allow rapid relocation of the active port within a compact enclosure,
effectively exploiting small-scale spatial channel fluctuations with
extremely low hardware overhead.
A more radical form of spatial reconfigurability is provided by movable
antennas (MA) 
\cite{ZhuMa2024MovableAntennas,ZhuMa2024ModelingandPerformance,ZhuMa2024MovableAntennaEnhanced,shao20246d},
which physically reposition a single RF element within a bounded
2D/3D/6D region.  
Unlike conventional fixed-geometry antennas (FGA), both FA and MA convert the
spatial position itself into a controllable communication resource,
enabling the system to probe or select favorable channel realizations
without increasing the number of RF chains, transmit power, or antenna
aperture.  
This spatial mobility provides several notable advantages:
(i) access to rich channel diversity at minimal hardware cost;
(ii) position-domain scheduling for improved signal-to-noise ratio (SNR) or lower correlation;
and (iii) channel-shaping capabilities comparable to MIMO or RIS but
with substantially lower energy consumption.  
Consequently, FA and MA, alongside RIS and MIMO, are emerging as key
enablers of 6G reconfigurable spatial communications, with promising
applications in integrated sensing and communication (ISAC), cooperative networks, and green wireless systems
\cite{Zou2024ShiftingtheISAC,New2025FluidAntennaSystems,ChenZhao2025AntennaPosition,CaoJiang2025JointAntennaPosition}.

Meanwhile, index modulation (IM) \cite{basar2016index} has emerged as an effective technique for improving spectral efficiency (SE) and reducing hardware cost by conveying information through an additional {index dimension}. The key idea of IM is to activate or select certain entities, such as antennas, subcarriers, time slots, carrier frequencies, or channel states \cite{Aydın2016ANovelSM,mao2018novelindexmodulation,ZhangZou2024RISAidedIndex}, thereby enabling the transmission of extra bits beyond conventional signal constellations. As a result, IM can enhance SE without increasing the modulation order, making it attractive for green communication systems with strict hardware and energy constraints.

IM has been widely integrated into fixed-antenna architectures. A representative example is spatial modulation (SM) \cite{mesleh2008spatial}, which conveys index bits via antenna activation. Numerous extensions, including generalized SM (GSM) \cite{DiHaas2014SpatialModulation}, quadrature SM (QSM) \cite{GuoZhang2019SignalShaping}, enhanced SM (ESM) \cite{Althunibat2018EnhancingSpatial}, and precoding SM (PSM) \cite{HeWang2018SpatialModulation}, further improve SE and reliability through richer activation patterns and advanced processing techniques. 
Beyond conventional MIMO systems, IM has also been combined with emerging architectures. For instance, frequency diverse arrays (FDA) \cite{HuangOrlando2025GLRTBased} introduce additional index dimensions across frequency and spatial domains \cite{HuangXu2025GeneralizedCode,jian2023mimo,Jian2024FDAMIMO,HuangShan2020IndexModulation}. In mmWave systems, spatial scattering modulation (SSM) \cite{ZhuChen2023QuadratureSpatial} exploits sparse propagation paths for index transmission, while its extension, polarized SSM (PSSM) \cite{LiKim2019PolarizedSpatial}, further leverages polarization diversity.

Despite these advances, most existing IM schemes rely on fixed-position antenna architectures (FPA-IM), whose index space is fundamentally constrained by the physical antenna layout, limiting scalability in the spatial domain. To overcome this limitation, recent works have explored IM with reconfigurable antenna structures such as RIS and FA \cite{LiBai2024RISBased}. In particular, FA-enabled IM schemes have been proposed for MIMO systems \cite{zhu2024index}, grouping-based designs \cite{guo2025fluid}, and position-domain index modulation \cite{yang2024position}, demonstrating improved SE and robustness. Furthermore, continuous-trajectory FA (CT-FA) based IM \cite{liu2025index} exploits channel variations induced by antenna motion to enable covert communication.

Despite the demonstrated potential of MA in enhancing 
spatial degrees of freedom and improving channel reconfigurability, to 
the best of the authors' knowledge, no existing work has systematically 
integrated the spatial mobility of MA with the IM
framework. In particular, the literature lacks a unified study on how to 
construct index states from the continuous two-dimensional movable region 
of an MA, nor does it provide corresponding modulation schemes, detection 
algorithms, or performance analyses tailored for such systems. This clear 
research gap presents a timely opportunity to explore the synergy between 
MA and IM.

Motivated by this observation, combining the spatial mobility of MA with the indexing principle of IM opens up an entirely new modulation dimension, where the physical position of the antenna itself serves as a controllable information-bearing resource. Such a hybrid architecture can enhance SE  without requiring additional RF chains or expanded bandwidth, while exploiting the rich channel diversity induced by MA movement.
To fully leverage this new modulation dimension, it is essential to carefully reconcile the discrete nature of IM with the intrinsically continuous spatial mobility of MA. In this regard, although IM introduces discrete index states, the proposed MA-IM framework does not quantize the antenna motion itself at a conceptual level. Instead, the intrinsic continuous spatial degrees of freedom of MA are preserved through dense spatial sampling and channel-domain design, which enables MA-IM to retain the fundamental characteristics of MA systems.
However, the integration of continuous spatial mobility with discrete index modulation gives rise to two fundamental challenges.

\begin{itemize}

\item \textbf{How to discretize the continuous movable region into a finite set of sampling points?}  
The MA operates over a continuous spatial region, whereas IM inherently relies on discrete states. Therefore, it is necessary to construct a sufficiently dense set of sampling points to approximate the continuous spatial domain, while preserving the underlying channel variations induced by antenna movement.

\item \textbf{How to select representative sampling points for reliable index modulation?}  
Although dense spatial sampling captures the continuous nature of MA, directly indexing all sampling points is generally unreliable due to strong spatial correlation, which leads to poor distinguishability among index states. Therefore, it is essential to select a subset of representative sampling points that maximizes channel separability and enables reliable index detection.

\end{itemize}

To overcome these challenges, this paper presents a novel MA-based
index modulation framework and makes the following key contributions:

\begin{itemize}

\item We develop a practical MA-IM system model that converts the
continuous movable region of the antenna into a discrete set of
candidate transmission positions through dense spatial sampling.
This discretization reveals a large pool of channel realizations
induced by the MA mobility and provides a foundation for constructing
index-modulated transmission states.

\item We introduce the concept of {anchor selection} for
MA systems. Unlike conventional FA-IM schemes that assume
a predefined set of antenna ports and only perform indexing among
them, the proposed framework first selects a subset of representative
ports (anchors) from a dense set of candidate MA sampling locations.
This additional design freedom allows the system to exploit the
underlying channel geometry and significantly improve the
separability among index states.

\item Based on the proposed anchor-selection framework, we develop a
family of MA-IM schemes with progressively increasing levels of
channel awareness, including random activation, geometry-based
selection, SNR-based representative selection, cell-constrained
max--min design, global max--min optimization, and joint
constellation-distance-aware anchor selection.
These schemes form a unified design hierarchy that systematically
enhances the separability of index states in the channel domain
and, ultimately, in the joint index--modulation constellation.
This framework enables a comprehensive performance comparison
across geometry-driven, channel-aware, and constellation-aware
port selection strategies.

\item For the resulting MA-IM transmission schemes, we derive joint
maximum-likelihood (ML) detectors that simultaneously estimate the
active index state and the transmitted QAM symbol. In addition, tight
union-bound expressions of the average bit error probability (ABEP)
are obtained for representative schemes, enabling analytical
performance characterization under various system parameters.

\item To reduce the computational burden of exhaustive ML detection,
we further design a low-complexity two-stage detector that achieves a
favorable performance–complexity tradeoff and makes MA-IM suitable for
practical implementation.

\item Extensive Monte Carlo simulations validate the accuracy of the
derived ABEP expressions and demonstrate that the proposed
channel-aware anchor selection strategies significantly improve the
reliability of MA-IM systems. The results also provide useful design
guidelines on sampling density, index alphabet size, and modulation
order for MA-enabled communication systems.

\end{itemize}

Overall, the results demonstrate that MA-IM can effectively exploit 
spatial reconfigurability to enhance modulation dimensionality and system 
performance without increasing hardware complexity. The proposed 
framework offers new insights and practical methodologies for the design 
of next-generation green, SE, and reconfigurable 
wireless communication systems.

The remainder of this paper is organized as follows. 
Section~\ref{sec2} establishes the overall MA-based transmission
framework. Besides, a family of anchor-selection
strategies for constructing index-modulated transmission states are given in Section~\ref{sec22}.
Section~\ref{sec3} presents the corresponding joint ML detectors for all
schemes and develops a low-complexity two-stage detection algorithm. 
Section~\ref{sec4} provides a comprehensive performance analysis,
including the achievable throughput and the derived ABEP upper bounds.
Section~\ref{sec:simulation} reports extensive simulation results that validate
the theoretical findings. 
Finally, Section~\ref{sec6} concludes the paper and outlines several
promising future research directions.

\section{System Model and Movable Antenna Discretization}
\label{sec2}
The MA-IM system investigated in this paper is depicted in Fig. \ref{fig1}. At the transmitter, a MA is employed, whereas the user adopts conventional FPA such as phase array (PA). The MA module is connected to the RF chain through flexible coaxial cables, thereby enabling real-time and continuous adjustment of its spatial position.
Let the position of an MA be represented by the Cartesian coordinate vector $\boldsymbol{t} = [x, y]^T \in \mathcal{C},$
where $\mathcal{C}$ denotes the designated two-dimensional (2D) region within which the transmit MA is allowed to move freely without mechanical constraints.
Throughout this work, we focus on narrowband and quasi-static channel conditions. We further assume that the MA can relocate with sufficiently high agility such that the positioning overhead is negligible relative to the extended channel coherence time, thus ensuring that MA reconfiguration does not disrupt communication \cite{mazhu2023mimocapacity}. 
% \begin{figure}[H]
%     \centering
%     \includegraphics[width=1\linewidth]{figure1.pdf}
%     \caption{Movable antenna architecture}
%     \label{fig:placeholder}
% \end{figure}
\begin{figure*}[htp]
	\centering
		{\includegraphics[width=0.68\textwidth]{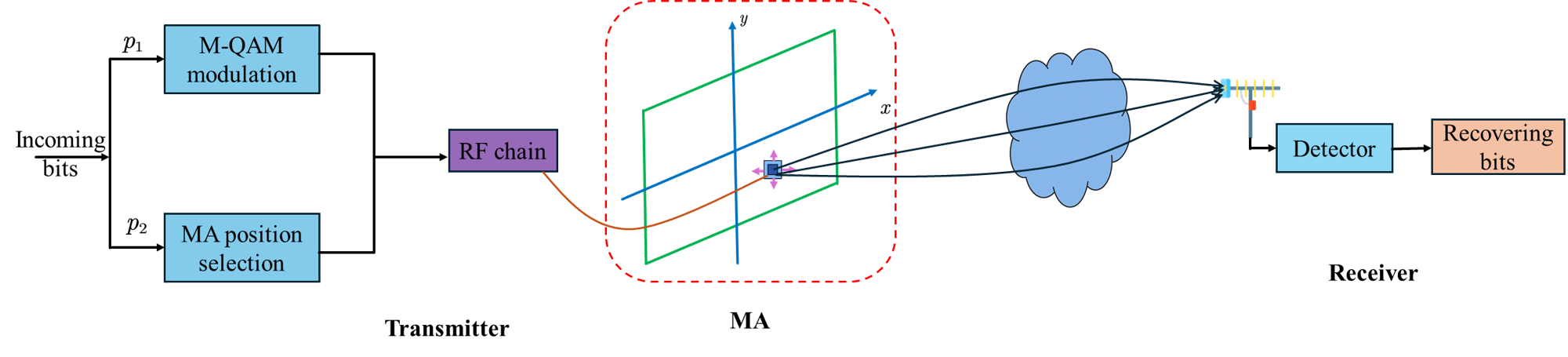}}	
		\caption{System model of the proposed MA-IM framework.}
		\label{fig1}	
\end{figure*}

\subsection{Channel model}
The channel vector between the MA-IM transmitter and the user is determined by both the propagation environment and the positions of the MA. Since the movement range of the antennas is negligible compared with the signal propagation distance, the far-field assumption is adopted. Based on this assumption, the MA–PA channels can be modeled using the plane-wave approximation \cite{zhuma2023modelingperformance}. Consequently, for each propagation path, the complex path coefficients at different MA locations exhibit different phases. This property allows the spatial variation of the channel to be fully characterized by the phase term induced by the propagation distance difference.

\begin{figure}[htp]
    \centering
    \includegraphics[width=0.7\linewidth]{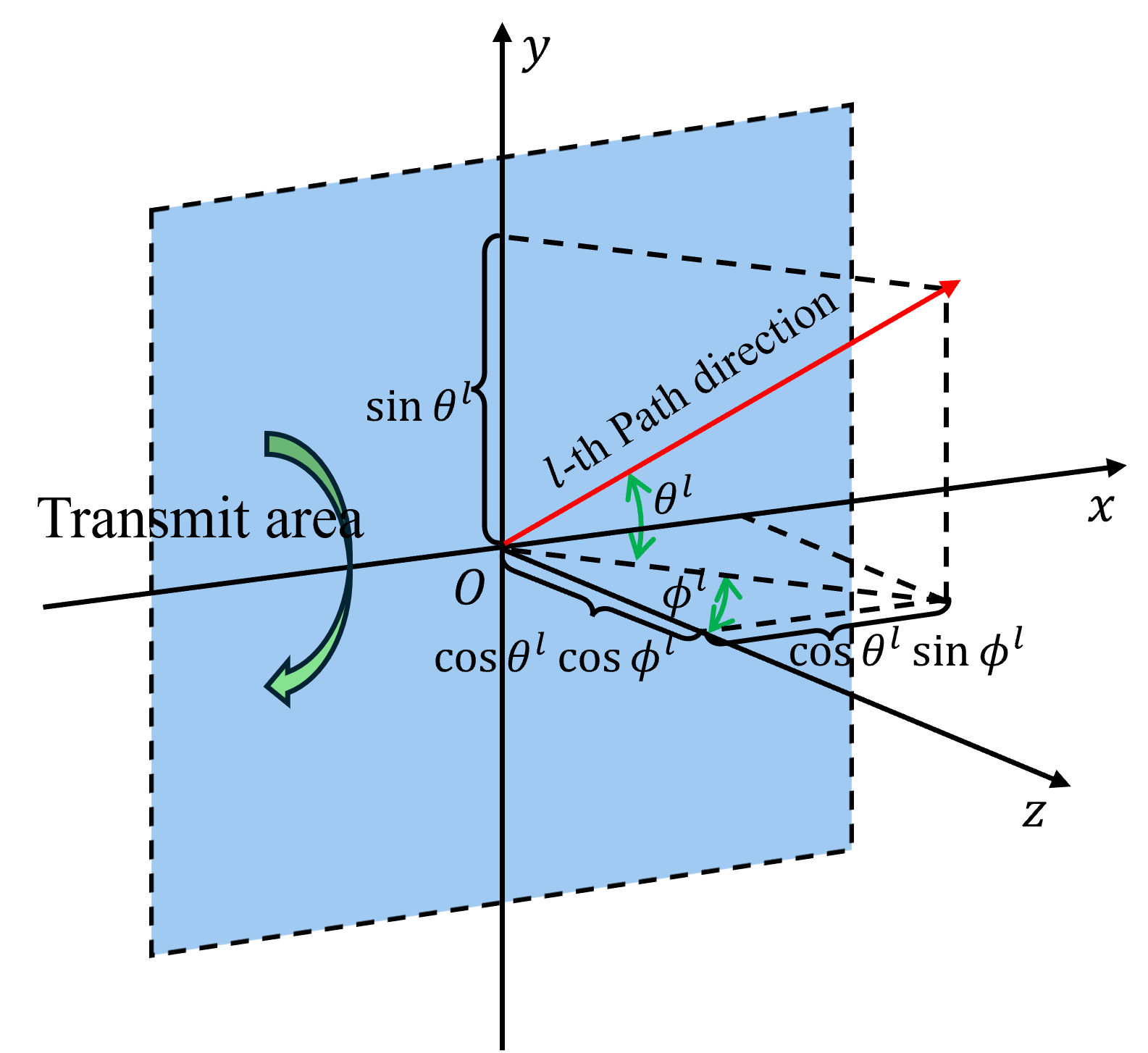}
    \caption{Illustration of the spatial angle domain corresponding to the transmit region}
    \label{fig:placeholder}
\end{figure}
Assume that the number of propagation paths is $L$.
As illustrated in Fig.~2, the $l$-th path is characterized by its
elevation and azimuth angles of departure (AoD), denoted by
$\theta^l$ and $\phi^l$, respectively. The corresponding 3D unit
propagation direction can be defined according to the geometric
coordinate system in Fig.~2. Since MA is restricted
to a 2D transmit region, only the projection of the propagation
direction onto this region is relevant for characterizing the phase
variation across different MA positions.

Accordingly, the effective 2D wavefront vector of the $l$-th path on
the transmit region is written as $
\boldsymbol{n}^l=\left[ \sin \theta ^l\cos \phi ^l,\cos \theta ^l \right] ^T
$.

Let the origin of the MA region be
$\boldsymbol{o}=[0,0]^T$, and let
$\boldsymbol{t}=[x,y]^T$ denote an arbitrary transmit position on the
2D region. Under the plane-wave assumption, the path-length
difference of the $l$-th propagation path between $\boldsymbol{t}$
and $\boldsymbol{o}$ is given by the projection of the displacement
$\boldsymbol{t}-\boldsymbol{o}$ onto the effective wavefront vector,
i.e.,
\begin{equation}
\rho^l(\boldsymbol{t})
=
(\boldsymbol{n}^l)^T(\boldsymbol{t}-\boldsymbol{o})
=
x\sin\theta^l\cos\phi^l+y\cos\theta^l.
\end{equation}

Thus, the field-response vector characterizing the transmit MA is defined as \cite{jiangZhang2025movableantenna}
\begin{equation}
    \label{}
\boldsymbol{f}\left( \boldsymbol{t} \right) =\left[ e^{j\frac{2\pi}{\lambda}\rho ^1\left( \boldsymbol{t} \right)},e^{j\frac{2\pi}{\lambda}\rho ^2\left( \boldsymbol{t} \right)},\cdots ,e^{j\frac{2\pi}{\lambda}\rho ^L\left( \boldsymbol{t} \right)} \right] ^T\in \mathbb{C} ^L,
\end{equation}
with the symbol $\lambda$ denoting the wavelength.

\subsection{Port discretization}
To facilitate theoretical analysis and performance evaluation, this paper introduces a sampled MA model. Instead of allowing the MA to move continuously over the region $\mathcal{C}$, we approximate this continuous area using a sufficiently dense set of discrete sampling points. Formally, the continuous movement space is discretized into a finite set $
\mathcal{C} _d=\left\{ \mathbf{t}_1,\mathbf{t}_2,\dots ,\mathbf{t}_Q \right\} 
$ with $Q$ being the total sampling number, and the MA is assumed to move among these discrete points. This discretization converts the original continuous-position selection problem into a finite-state switching problem, significantly simplifying the modeling and analysis.

As long as the sampling points are dense enough, the discrete points can accurately capture the spatial phase variations and channel fluctuations over the continuous region. Therefore, the sampled points serve as a high-fidelity approximation of the continuous MA movement space, without sacrificing the spatial degrees of freedom or channel diversity gains inherent to MA.

Moreover, this sampled MA abstraction is also aligned with practical hardware implementations. Whether the MA is driven by mechanical actuators (e.g., linear rails or stepper motors) or realized through micro-electro-mechanical systems (MEMS)/liquid-metal mechanisms as in FA, the physically reachable antenna positions are inherently discrete with finite spatial resolution. Furthermore, under the sampling-theoretic constraints, Lemma \ref{lem1} provides the maximum allowable sampling interval for discretizing the movable-antenna region, thereby offering a theoretical guideline for constructing a finite set of sampling points from the original continuous movement space. 

\begin{lemma}
\label{lem1}
Consider a rich-scattering environment where the spatial correlation between 
two antenna positions separated by distance $\Delta r$ is approximated by \cite{jakes1994microwave,New2024AnInformation}
\begin{equation}
\label{eq3}
    \rho(\Delta r)=J_{0}\!\left(\frac{2\pi \Delta r}{\lambda}\right),
\end{equation}
where $J_0(\cdot)$ denotes the zeroth-order Bessel function.
Given a target correlation level $\rho_{\mathrm{tar}}\in[0,1)$, 
the maximum spacing $\Delta r_{\max}$ that guarantees 
$\rho(\Delta r)\ge \rho_{\mathrm{tar}}$ is
\begin{equation}
    \Delta r_{\max}\approx \frac{\lambda}{\pi}\sqrt{1-\rho_{\mathrm{tar}}}.
\end{equation}

\end{lemma}

\begin{proof}
Using the approximation $J_0(x)\approx 1-x^2/4$ for small $x$, we obtain
\begin{equation}
    \rho(\Delta r)\approx 
    1-\frac{1}{4}\!\left(\frac{2\pi\Delta r}{\lambda}\right)^{2}.
\end{equation}
Imposing $\rho(\Delta r)\ge \rho_{\mathrm{tar}}$ leads to
\begin{equation}
    1-\frac{1}{4}\!\left(\frac{2\pi\Delta r}{\lambda}\right)^{2}
    \ge \rho_{\mathrm{tar}},
\end{equation}
which gives
\begin{equation}
    \Delta r_{\max}\approx 
    \frac{\lambda}{\pi}\sqrt{\,1-\rho_{\mathrm{tar}}\,}.
\end{equation}

% For a 2D square grid, the largest Euclidean spacing occurs along the diagonal,
% i.e., $\sqrt{2}\,dx$. Requiring $\sqrt{2}\,dx\le\Delta r_{\max}$ yields
% \begin{equation}
%     dx \le \frac{\lambda}{\pi\sqrt{2}}\sqrt{1-\rho_{\mathrm{tar}}}.
% \end{equation}
This completes the proof.
\end{proof}

Compared with the Nyquist sampling limit $\lambda/2$\footnote{
As indicated by \eqref{eq3}, when the sampling interval is chosen as $\lambda/2$, the channel responses at two adjacent sampling points become almost uncorrelated. This, however, is not the intended design principle of this work. Our goal is to ensure that any two neighboring sampling points remain highly correlated even after discretization.
}, the spacing $\Delta r_{\max}$ derived in Lemma 1 based on channel-similarity constraints is strictly smaller. Consequently, the admissible sampling intervals for discretizing the movable-antenna region along the $x$– and $y$–directions can be written as:
\begin{equation}
    \label{eq8}
dx\in \left[ 0,\Delta r_{\max} \right],\quad dy\in \left[ 0,\Delta r_{\max} \right].
\end{equation}
% $
% dx\in \left( 0,\min \left\{ \Delta r_{\max},\frac{\lambda}{2} \right\} \right] ,dy\in \left( 0,\min \left\{ \Delta r_{\max},\frac{\lambda}{2} \right\} \right] 
% $.

(8) indicates that as the sampling interval approaches zero, the movable-antenna region effectively behaves as a continuous domain. However, to ensure the target channel similarity $\rho_{\mathrm{tar}}$, the sampling interval must be bounded above by $\Delta r_{\max}$. In the subsequent  section, to simplify the analysis and highlight the key performance trends, we set $d_x$ and $d_y$ to their maximum allowable values.

\begin{figure}[htp]
	\centering
		{\includegraphics[width=0.48\textwidth]{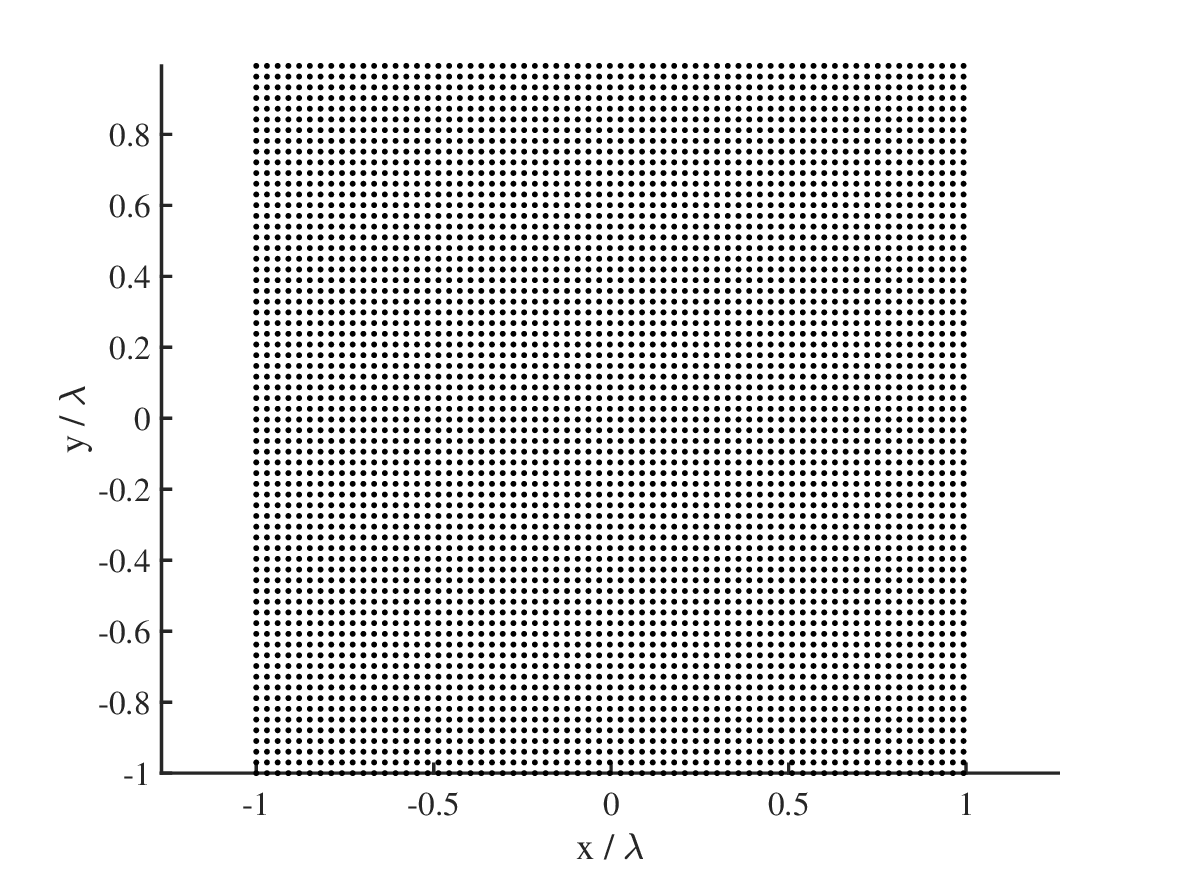}}
		% \subfigure[]{\includegraphics[width=0.35\textwidth]{results/FDAvsMIMO_L2_K16.eps}}
		% \subfigure[]{\includegraphics[width=0.35\textwidth]{results/FDAvsMIMO_L1_K32.eps}}		
		\caption{Illustration of the sampling layout in the MA region under the target channel similarity constraint $\rho_{\mathrm{tar}} = 0.9$.}
		\label{fig4}	
\end{figure}

Fig. \ref{fig4} shows all sampling points within the MA’s movable region under the design target $\rho_{\mathrm{tar}} = 0.9$. 
The channel similarity between any two positions $\boldsymbol{t}_m$ and $\boldsymbol{t}_n$ is defined as
\begin{equation}
    \rho_h = \boldsymbol{f}(\boldsymbol{t}_m)^{\dagger}\boldsymbol{f}(\boldsymbol{t}_n)/L .
\end{equation}
By computing the similarity for all adjacent sampling-point pairs in Fig. \ref{fig4} and taking the average, an averaged similarity of $\rho_h = 0.8945$ is obtained, indicating that the corresponding channels are highly correlated.
This observation further confirms the validity of the proposed lemma.

\section{Index modulation scheme}
\label{sec22}
Building upon the above discretization and partitioning framework,
a straightforward approach is to directly treat all sampled MA
positions as candidates for index modulation. From an
information-theoretic perspective, such a design maximizes the
number of index states and thus increases the achievable index
modulation rate. However, directly indexing all sampled positions
is generally undesirable in practice. As the indexed positions
become denser, the spatial separation between adjacent positions
decreases, and the resulting channel responses become increasingly
similar. This reduces the distinguishability among index states in
the channel domain and makes reliable detection at the receiver
more challenging, thereby leading to a higher bit error rate (BER).

Therefore, it is essential to select a subset of representative
sampling points for IM. The objective is to enhance
the distinguishability among index states and improve the reliability
of ML detection. The primary focus of this work is on how to design
such representative anchor selection strategies\footnote{The joint
design of the number of index states and the corresponding anchor
selection, particularly under practical constraints such as channel
correlation and computational complexity, remains an open problem
and is left for future work.}.
Hence, we develop a family of MA-IM schemes that differ in their
representative port selection strategies.

As a baseline, a random activation scheme is considered,
where the MA can move freely within the selected mobility
block. A geometry-based scheme then fixes the MA at the
geometric center of each block, providing a deterministic
representative position. To further exploit channel
variations, an SNR-based scheme selects the port with the
largest channel gain within each block. In addition, a distance-based scheme employs a max--min
(i.e., maximizing the minimum pairwise channel distance,
which is sometimes referred to as a min--max formulation
in related literature \cite{Boyd2004Convex}) criterion to select the representative
port that maximizes the minimum channel distance among
different index states within each block.
Next, we consider a channel-domain anchor design that
removes the geometric partition constraint and directly
selects representative ports according to a global max--min
channel-distance criterion.
Furthermore, to fully capture the joint effect of index
selection and signal modulation, we propose a
constellation-aware anchor design, which directly optimizes
the minimum Euclidean distance of the resulting joint
index--modulation constellation, thereby achieving
significantly improved detection performance.

\subsection{Random activation (Scheme 1)}

In this scheme, we first divide the 2D movable region of the MA into
$N = N_1 \times N_2,$
where \(N_1\) and \(N_2\) denote the numbers of grids along the horizontal and vertical directions, respectively. 
To ensure distinct MA-IM states, the condition
$N < Q$
must be satisfied. 
This raises a practical question: when a grid index is selected for transmission, which sampling point within that grid should be used as the actual MA position?

The movable region is partitioned into $N$ geometric cells, and the $c$-th cell contains a
set of candidate sampling points denoted by $\mathcal{S}_c$.
In this scheme, the transmit position is randomly selected from the
candidate points inside the activated cell, i.e.,
\begin{equation}
  q_{\mathrm{tx}} \sim \mathrm{Unif}\big(\mathcal{S}_c\big),
  \label{eq:qtx_random}
\end{equation}
and the MA transmits from the corresponding physical coordinate $\boldsymbol{t}_{q_{\mathrm{tx}}}$ with $\mathrm{Unif}(\cdot)$ denoting a uniform random selection from the specified set \footnote{Compared with the other fixed-center transmission strategies
considered in this paper, the random sampling in
\eqref{eq:qtx_random} better reflects the continuous mobility
characteristic of MA. As will be shown in
the simulation results, this scheme is generally unreliable
in practice due to its poor detection reliability. Nevertheless, it
provides a realistic lower performance bound and serves as a
useful benchmark for evaluating the effectiveness of the
proposed MA-based index modulation schemes.}. 

Let 
$\boldsymbol{h}_q 
\triangleq 
\boldsymbol{f}(\boldsymbol{t}_q)
\in \mathbb{C}^{N_r}$
be the corresponding channel vectors. Consequently, the received signal $\boldsymbol{y}\in\mathbb{C}^{N_r\times 1}$ can be written as
\begin{equation}
\boldsymbol{y}
=
E_s s\,\boldsymbol{h}_{q_{\mathrm{tx}}}
+
\boldsymbol{n}.
\end{equation}
 Here, $E_s$ denotes the transmit energy per symbol and
$\boldsymbol{n}\in\mathbb{C}^{N_r\times 1}$ is the complex additive
white Gaussian noise (AWGN), satisfying $\boldsymbol{n}\sim\mathcal{CN}(0,N_0\boldsymbol{I})$. Moreover, $s \in \mathcal{S}_M$ denotes the transmitted modulation symbol drawn from an $M$-QAM constellation, where $|\mathcal{S}_M| = M$ and $\mathbb{E}[|s|^2] = 1$.

Although Scheme~1 reflects the random mobility of the movable antenna,
the unknown transmit position within each cell leads to unreliable
demodulation and degraded BER performance. Therefore, it mainly serves
as a baseline to motivate anchor-based designs, which aim to improve
the reliability of index transmission.

\subsection{Geometry-based anchoring (Scheme 2)}
To address the randomness and unreliability of Scheme~1,
a natural idea is to assign a deterministic representative
transmit position to each cell. A simple yet intuitive
choice is to exploit the geometric structure of the partition. In Scheme~2, each cell is represented by a fixed anchor point
determined based on the geometric structure of the partition.

Let $\mathbf c_c$ denote the geometric center of the $c$-th cell and let $\mathbf t_q$ denote the coordinate of the $q$-th sampling point. The representative port of cell $c$ is selected as the sampling point closest to the cell center. Hence,
the representative anchor point \(a_c\) is chosen as
\begin{equation}
    \label{}
a_c=\arg\min_{q\in \mathcal{S}_c} \bigl\| \boldsymbol{t}_q-\boldsymbol{c}_c \bigr\|^2
\end{equation}

Thus, each grid is represented by a unique anchor point, which serves as the actual MA transmit location for that grid, as shown in Fig.~\ref{fig5}.
\begin{figure}[htp]
	\centering
		{\includegraphics[width=0.4\textwidth]{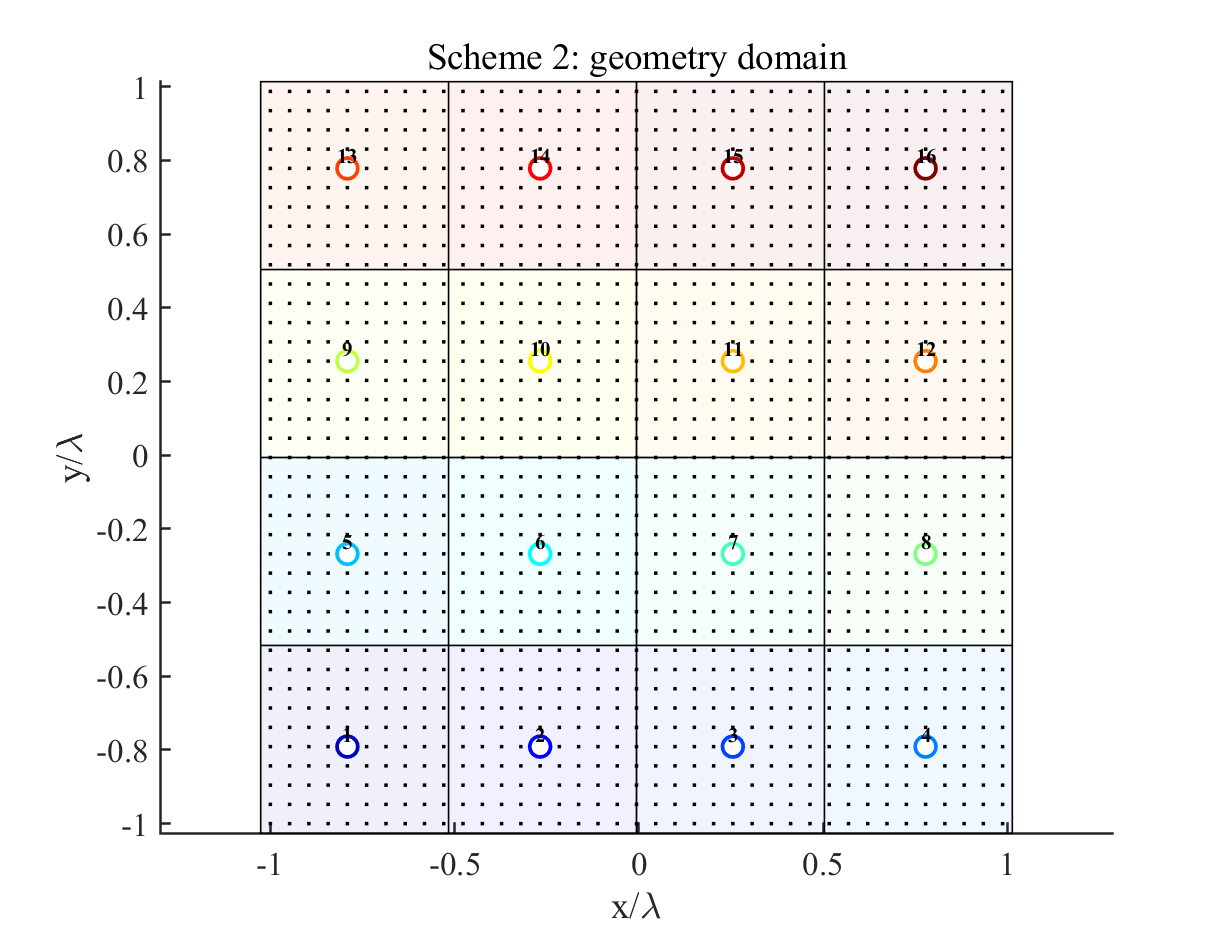}}	
		\caption{Example of MA indexing: each block corresponds to one index value, and the colored point denotes the selected physical transmit coordinate for that index.}
		\label{fig5}	
\end{figure}

Consequently, IM operates directly on this set of anchors.  
Based on Fig.~\ref{fig1}, the incoming bitstream can therefore be partitioned into two sub-vectors of lengths
$\log_2 M \quad $ {and} $\left\lfloor \log_2 N \right\rfloor$
denoted as \(\boldsymbol{p}_1\) and \(\boldsymbol{p}_2\). 
The bits in \(\boldsymbol{p}_1\) are mapped to an \(M\)-QAM symbol \(s\), whereas the IM bits in \(\boldsymbol{p}_2\) are used to select a unique MA position pattern
$I = \{1, 2, \ldots, N\}.$

Let $\mathcal{A}=\{a_1,a_2,\ldots,a_C\}$ denote the set of
anchor indices, where $a_c$ corresponds to the representative
sampling point selected in the $c$-th cell. When the index
state $c$ is activated, the MA directly transmits from the
anchor position $\boldsymbol{t}_{a_c}$.

The received signal
$\boldsymbol{y}\in\mathbb{C}^{N_r}$ can therefore be written as
\begin{equation}
\boldsymbol{y}
=
E_s s\,\boldsymbol{h}_{a_c}
+
\boldsymbol{n}.
\end{equation}

While Scheme~2 assigns a deterministic anchor to each cell
based purely on geometric proximity, the resulting anchor
positions do not necessarily correspond to favorable channel
conditions. In practical wireless environments, the channel
strength may vary significantly across different sampling
points within the same cell due to multipath propagation and
small-scale fading. As a result, selecting the geometric center
as the representative position may lead to suboptimal
transmission reliability.

To address this limitation, we next consider a channel-aware
representative selection strategy that exploits the
instantaneous channel strength within each cell.

\subsection{SNR-based selection (Scheme 3)}

To improve the reliability of index transmission while preserving the geometric partition structure, we consider a representative port selection strategy based on the instantaneous channel strength within each cell. This scheme is referred to as {Scheme~3}.

Hence, the representative port for the $c$-th cell is selected as the port with the largest channel gain within that cell, i.e.,
\begin{equation}
q_{c}^{\star}=\mathrm{arg}\max_{q\in \mathcal{S} _c} |\boldsymbol{h}_q|^2.
\label{eq:snr_selection}
\end{equation}
To facilitate practical implementation, the SNR-based
representative selection procedure is summarized in
Algorithm~\ref{alg:scheme3}. The algorithm simply scans all
candidate ports within each cell and selects the port with
the maximum channel gain as the representative position.

The selected representative ports form the set
\[
\mathcal{B}
=
\{q_1^\star,q_2^\star,\dots,q_C^\star\}.
\]

The effective channel corresponding to the $c$-th index state is therefore $
\tilde{\boldsymbol{h}}_c=\boldsymbol{h}_{q_{c}^{\star}}.
$

Compared with the geometry-based scheme (Scheme~2), Scheme~3 improves the received signal quality by selecting the strongest port in each cell, thereby increasing the effective SNR for each index state.

When the $c$-th index state is activated, the transmitted signal experiences the effective channel $\tilde {\boldsymbol{h}}_c$. The received signal can be expressed as
\begin{equation}
\boldsymbol{y}=E_ss\tilde{\boldsymbol{h}}_c+\boldsymbol{n}.
\end{equation}

\begin{algorithm}[htp]
\caption{Cell-wise SNR-Based Representative Selection (Scheme 3)}
\label{alg:scheme3}
\begin{algorithmic}[1]
\REQUIRE Candidate port sets $\{\mathcal{S}_c\}_{c=1}^{C}$,
channel coefficients $\{h_q\}$
\ENSURE Representative ports $\{q_c^\star\}_{c=1}^{C}$

\FOR{$c = 1$ to $C$}
    \STATE
    $
    q_c^\star
    =
    \arg\max_{q \in \mathcal{A}_c} |h_q|^2
    $
\ENDFOR

\STATE return $\{q_c^\star\}$
\end{algorithmic}
\end{algorithm}

\begin{figure}[htp]
	\centering
		{\includegraphics[width=0.4\textwidth]{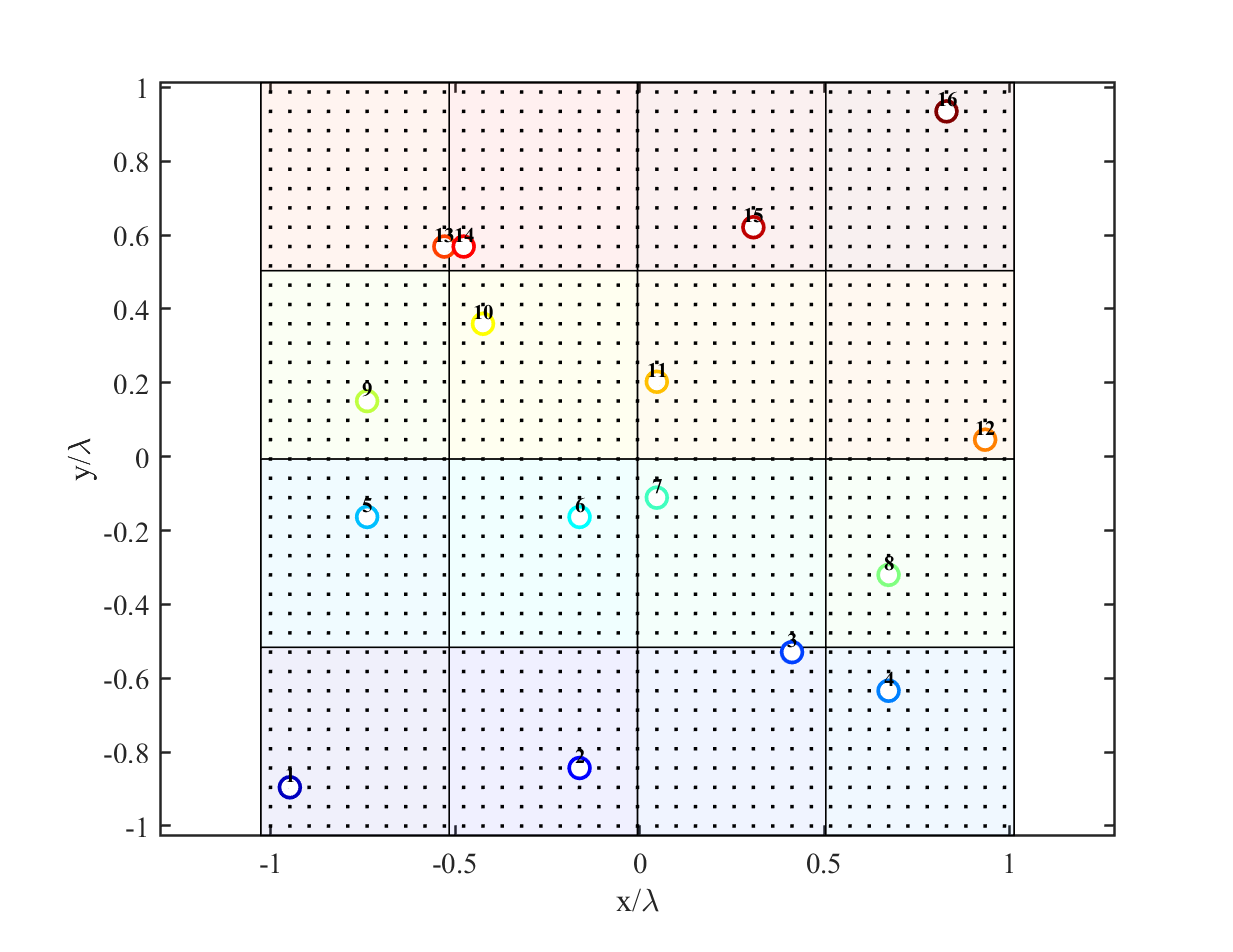}}	
		\caption{Example of MA indexing under Scheme~3. Black dots denote
candidate sampling positions, while colored circles indicate the
selected representative ports in each cell.}
		\label{figu6}	
\end{figure}

Fig.~\ref{fig8} illustrates the representative ports selected under
Scheme~3. The MA region is first partitioned into multiple
geometric cells, where each cell corresponds to one index state. Within
each cell, the sampling point with the strongest channel gain is chosen
as the representative transmit position according to
\eqref{eq:snr_selection}. The black dots denote all candidate sampling
positions, while the colored circles indicate the selected
representative ports. Compared with geometry-based anchoring, this
channel-aware selection strategy improves the effective SNR associated
with each index state and therefore enhances the reliability of index
detection.

While Scheme~3 selects the representative port in each cell
based on the strongest channel gain, this criterion focuses on
improving the signal strength of individual index states rather
than their mutual distinguishability. As a result, the selected
ports from different cells may still produce highly correlated
channel responses, which limits the separability of the
resulting index states.

To address this limitation, we next consider a representative
selection strategy that explicitly maximizes the minimum
pairwise channel distance among the selected ports.

\subsection{Distance-based selection (Scheme 4)}

To further improve the distinguishability among different
index states while preserving the geometric partition
structure, we consider a cell-constrained max--min
representative selection strategy, referred to as
{Scheme~4}.

Unlike Scheme~3, which selects the representative port
in each cell based on the strongest channel gain, the proposed
Scheme~4 explicitly considers the separability among
different index states. Specifically, one representative port
is deterministically selected from each cell such that the
minimum pairwise channel distance among all selected ports
is maximized. In this respect, the representative ports are obtained by solving
\begin{equation}
\{q_1^\star,\dots,q_C^\star\}
=
\arg\max_{q_1\in\mathcal{S}_1,\dots,q_C\in\mathcal{S}_C}
\min_{i\neq j}|\boldsymbol{h}_{q_i}-\boldsymbol{h}_{q_j}|.
\label{eq:scheme2mm}
\end{equation}

The optimization in \eqref{eq:scheme2mm} ensures that the channel states associated with different index positions are well separated in the complex channel space, thereby improving the distinguishability of different index states.

However, directly solving the combinatorial optimization in
\eqref{eq:scheme2mm} is computationally prohibitive when the
number of candidate ports is large. To address this issue, we
propose a greedy cell-constrained search algorithm to compute
the representative ports $\{q_1^\star,\dots,q_C^\star\}$.
The procedure is summarized in Algorithm~\ref{alg:scheme2mm}.

\begin{algorithm}[htp]
\caption{Greedy Cell-Constrained Max--Min Representative Selection (Scheme 4)}
\label{alg:scheme2mm}
\begin{algorithmic}[1]
\REQUIRE Candidate port sets $\{\mathcal{S}_c\}_{c=1}^{C}$, channel coefficients $\{\boldsymbol{h}_q\}_{q=1}^{Q}$, maximum iteration number $I_{\max}$
\ENSURE Representative ports $\{q_c^\star\}_{c=1}^{C}$

\STATE Initialize representative ports $q_c^{(0)}\in\mathcal{S}_c$ using geometric anchors
\STATE Set iteration index $i=0$

\REPEAT
    \FOR{$c=1$ to $C$}
        \STATE Fix representatives of other cells $\{q_j\}_{j\neq c}$
        \FOR{each candidate $q\in\mathcal{S}_c$}
            \STATE Compute distance metric
            \[
            \eta_c(q)=\min_{j\neq c}|\boldsymbol{h}_q-\boldsymbol{h}_{q_j}|
            \]
        \ENDFOR
        \STATE Update representative
        \[
        q_c^{(i+1)}=
        \arg\max_{q\in\mathcal{S}_c}\eta_c(q)
        \]
    \ENDFOR
    \STATE $i\leftarrow i+1$
\UNTIL{convergence or $i=I_{\max}$}

\STATE $q_c^\star = q_c^{(i)}$
\end{algorithmic}
\end{algorithm}

Once the representative ports set  $
\mathcal{C}=\{q_{c}^{\star}\}_{c=1}^{C}
$ are determined, the effective channel corresponding to the $c$-th index state is given by
$
\bar{\boldsymbol{h}}_c=\boldsymbol{h}_{q_{c}^{\star}}.
$

The received signal can be written as
\begin{equation}
\boldsymbol{y}=E_ss\bar{\boldsymbol{h}}_c+\boldsymbol{n}.
\end{equation}

\begin{figure}[htp]
	\centering
		{\includegraphics[width=0.4\textwidth]{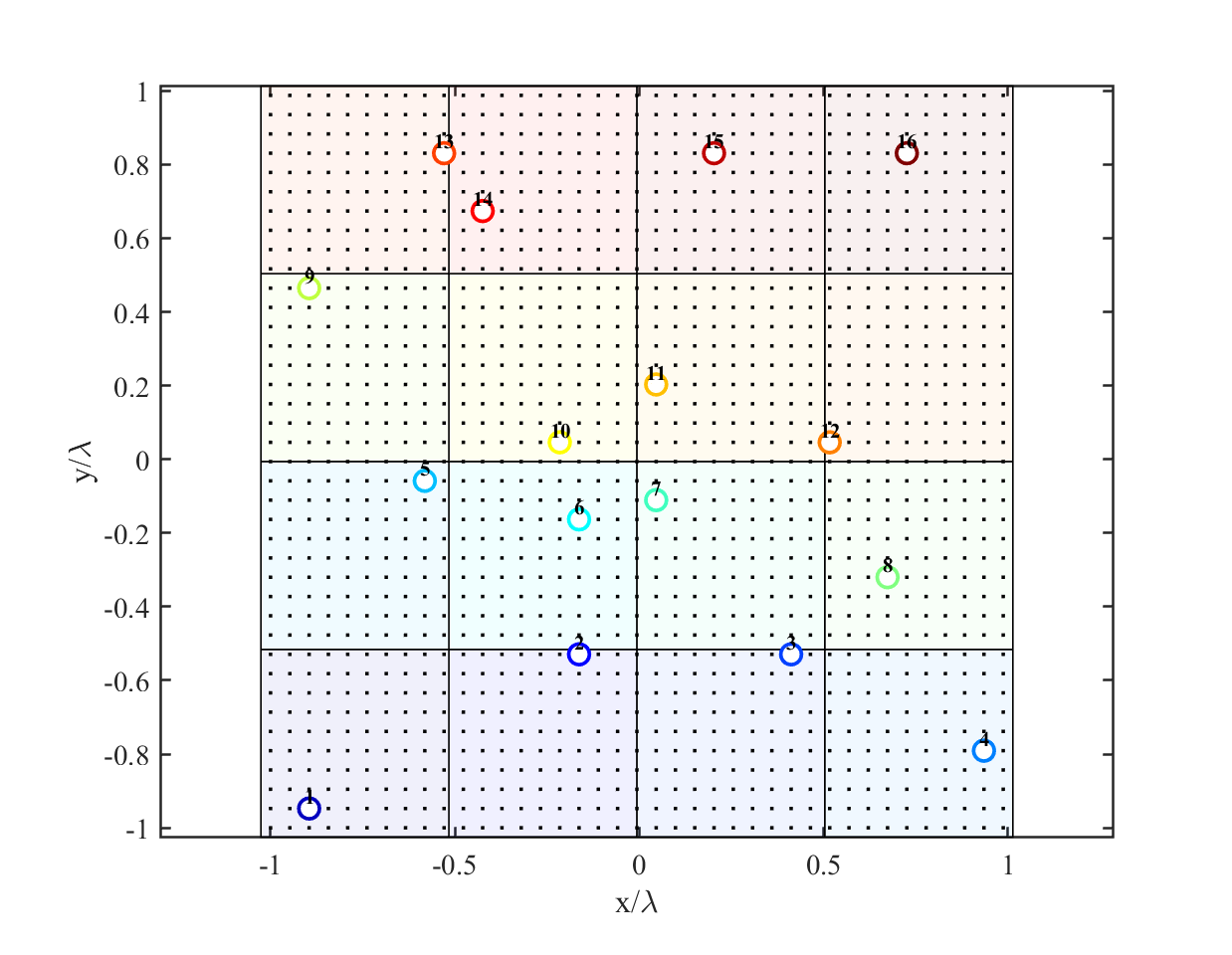}}	
		\caption{Example of MA indexing under Scheme~4. Black dots denote
candidate sampling positions, while colored circles indicate the
selected representative ports in each cell.}
		\label{figu7}	
\end{figure}

Fig.~\ref{figu7} shows the representative ports obtained by
Scheme~4. Compared with Fig.~\ref{figu6},
most representative ports are different because Scheme~4
optimizes the minimum channel distance among index states
rather than the individual channel strength. Nevertheless,
a few representative ports coincide with those in Scheme~3,
indicating that the strongest channel position may also yield
good channel separability in certain cells.

The previous schemes (Schemes~1--4) share a common design
principle in that the movable region is first partitioned into
geometric cells and one representative port is selected from
each cell. Although different criteria are used for
representative selection (random sampling, geometric anchoring,
SNR maximization, and max--min channel separation), the
underlying geometric partition constraint remains unchanged.

However, the wireless channel induced by MA may
vary highly non-uniformly across space, and spatial proximity
does not necessarily imply channel similarity. Consequently,
restricting the anchor positions to predefined geometric cells
may limit the achievable channel separability among index
states.

Motivated by this observation, we next consider a
channel-domain anchor design strategy that removes the
geometric partition constraint and directly optimizes the
anchor positions according to channel dissimilarity.

\subsection{Channel-domain max--min anchor design (Scheme 5)}

Unlike Schemes~1--4, which follow a cell-constrained design
based on geometric partitioning of the movable region,
Scheme~5 constructs the index states directly in the channel
domain. The key motivation is that MA-induced channels may
vary highly non-uniformly across space, whereas geometric
partitioning imposes a uniform spatial grid that does not
necessarily reflect the underlying channel structure.

Consequently, spatially adjacent cells in Schemes~1--4 may
still correspond to highly similar channel responses. In such
cases, the minimum separation between anchor channels is
essentially limited by the grid resolution. When the number of
index states $N$ increases, the geometric cells must shrink
accordingly, which inevitably reduces the minimum channel
distance between neighboring anchors and leads to a denser
joint constellation, thereby degrading the reliability of
index detection.

In contrast, Scheme~5 does not rely on geometric proximity.  
Instead, it explicitly selects anchor positions based on
{channel dissimilarity} and aims to {maximize the minimum
pairwise channel distance} among the chosen anchors.  
This max--min criterion ensures that the resulting anchor channels are
as mutually distinguishable as possible, yielding a joint constellation
with significantly larger minimum Euclidean distance and thereby
substantially more robust ML demodulation.

We seek an anchor set
\begin{equation}
    \label{}
    \mathcal{D}
=
\{a_1,a_2,\ldots,a_N\} 
\subseteq 
\{1,2,\ldots,Q\},
\end{equation}
that maximizes the minimum channel separation, namely
\begin{equation}
\label{eq:maxmin}
\max_{\mathcal{D}:\,|\mathcal{D}|=N}
\;\min_{i\neq j}
\big\|\boldsymbol{h}_{a_i} - \boldsymbol{h}_{a_j}\big\|.
\end{equation}
This max--min criterion is aligned with classical codebook optimization
and directly improves the minimum distance of the joint constellation,
which in turn lowers the ABEP union bound and the BER.

Since the combinatorial problem \eqref{eq:maxmin} is NP-hard \cite{gonzalez1985clustering}, we employ
the farthest-point sampling (FPS) strategy, which provides an effective
greedy approximation \cite{vazirani2001approximation}.  
The procedure is summarized in Algorithm~\ref{alg:farthest_point_anchor}.

\begin{algorithm}[t]
  \caption{Max--min channel-domain anchor selection via farthest-point sampling}
  \label{alg:farthest_point_anchor}
  \begin{algorithmic}[1]
    \REQUIRE Complex channel samples $\{\boldsymbol{h}_q\}_{q=1}^Q$, target number of anchors $N$
    \ENSURE Anchor index set $\mathcal{D} = \{a_1,\ldots,a_N\}$

    \STATE \textbf{ Step 1: Pre-compute pairwise channel distances}
    \FOR{$q_1 = 1$ \TO $Q$}
      \FOR{$q_2 = 1$ \TO $Q$}
        \STATE $D(q_1,q_2) = |\boldsymbol{h}_{q_1} - \boldsymbol{h}_{q_2}|$
      \ENDFOR
    \ENDFOR

    \STATE \textbf{ Step 2: Initialize the first anchor}
    \STATE Compute the average distance of each point:
           $\bar d(q) = \frac{1}{Q-1} \sum_{q' \neq q} D(q,q')$
    \STATE Select the point with the largest average distance:
           $a_1 = \arg\max_{q} \bar d(q)$
    \STATE $\mathcal{D} \gets \{a_1\}$

    \STATE \textbf{ Step 3: Iteratively select the remaining $C-1$ anchors}
    \FOR{$c = 2$ \TO $N$}
        \FOR{each $q \in \{1,\ldots,Q\}$}
           \STATE Compute the distance to current anchor set:
           \[
              d_{\min}(q) = \min_{a \in \mathcal{D}} D(q,a)
           \]
        \ENDFOR
        \STATE Exclude already selected anchors by setting
              $d_{\min}(a)=0$ for all $a\in\mathcal{D}$
        \STATE Select the farthest point: 
              \[
                a_c = \arg\max_{q} d_{\min}(q)
              \]
        \STATE Update: $\mathcal{D} \gets \mathcal{D} \cup \{a_c\}$
    \ENDFOR

    \STATE \textbf{return} $\mathcal{D}$
  \end{algorithmic}
\end{algorithm}

Furthermore, Fig.~\ref{fig66} shows the anchors obtained by
Algorithm~\ref{alg:farthest_point_anchor}. Compared with the
geometry-driven distribution in Fig.~\ref{fig5}, the anchors
selected by Scheme~5 are determined entirely in the channel
domain and are therefore not constrained by the geometric
grid. As a result, the spatial spacing among anchors is
adaptively adjusted according to channel dissimilarity.
This can be clearly observed in Fig.~\ref{fig66}, where
geometrically close positions may correspond to highly
distinct channel realizations, while the smallest channel
distance may occur between positions that are not spatially
adjacent. Such behavior highlights the benefit of channel-
domain anchor design in enlarging the minimum channel
separation among index states.
\begin{figure}[htp]
	\centering
		{\includegraphics[width=0.4\textwidth]{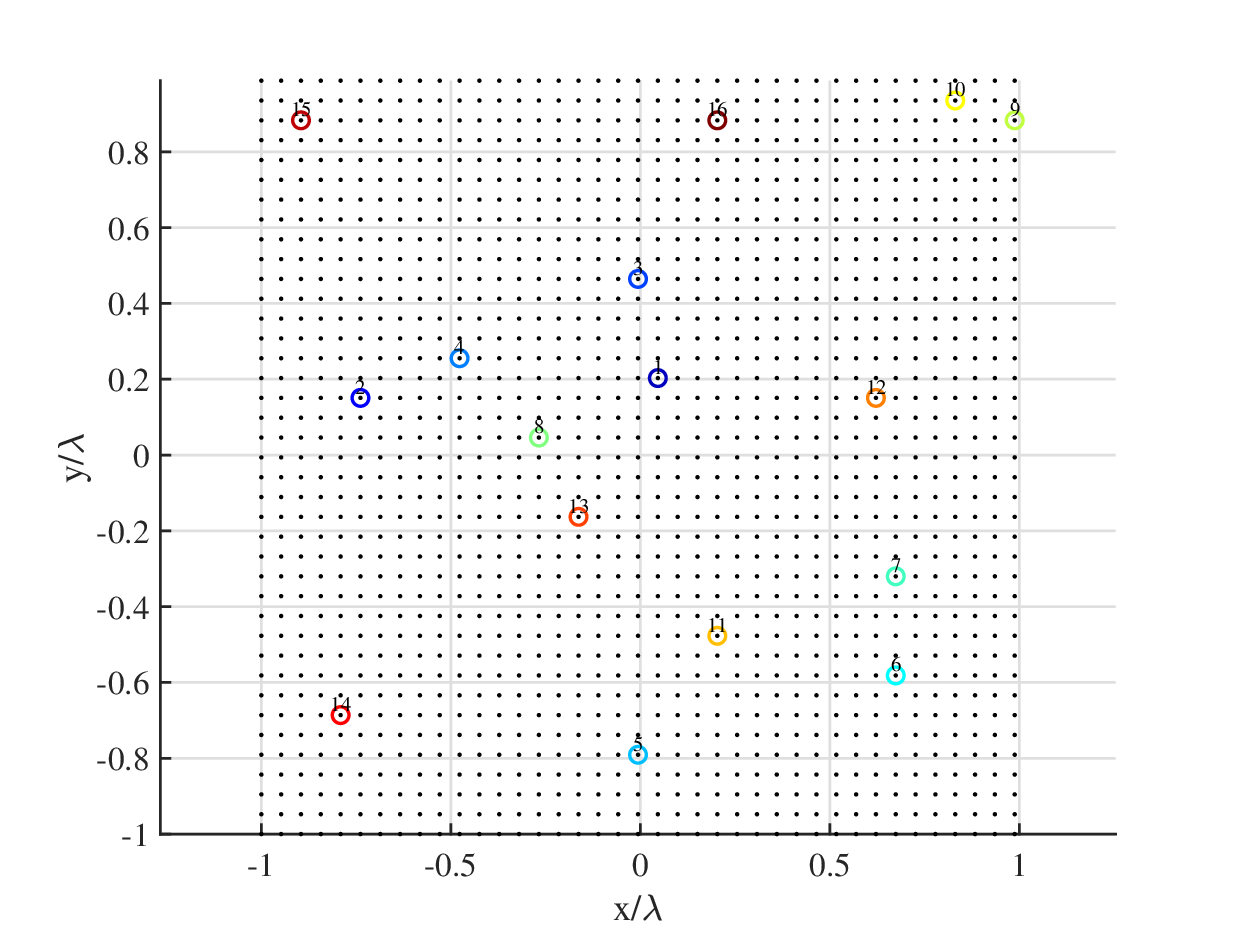}}	
		\caption{Geometry-domain visualization of the max--min anchor selection. Black dots represent all sampled MA positions, while colored circles 
indicate the selected $N$ anchors. }
		\label{fig66}	
\end{figure}

Once the anchor set $\mathcal{D}$ is obtained, the MA activates the
position $\boldsymbol{t}_{a_I}$ corresponding to the selected index $I$
and transmits symbol $s$.  
The received signal is then
\begin{equation}
\boldsymbol{y}
= E_s s\,\boldsymbol{h}_{a_I} + \boldsymbol{n}.
\end{equation}

For the anchor set $\mathcal{D}$, we define the minimum channel
separation as
\begin{equation}
d_{\min}^{h}
=
\min_{i\neq j}
\big\|\boldsymbol{h}_{a_i}-\boldsymbol{h}_{a_j}\big\|.
\end{equation}

Although the overall ABEP of MA-IM with QAM is determined by the
Euclidean distances of the joint constellation points
$\{\boldsymbol{h}_{a_i}s_m\}$, the channel-side metric
$d_{\min}^{h}$ provides a useful measure of the distinguishability
between index states.

In particular, when the same QAM symbol is transmitted from two
different anchors, the corresponding pairwise distance becomes
$|s_m|\,\|\boldsymbol{h}_{a_i}-\boldsymbol{h}_{a_j}\|$. Therefore,
increasing $d_{\min}^{h}$ generally improves the separability of
index states and reduces the probability of index-detection errors.
However, the overall BER performance ultimately depends on the
geometry of the joint constellation
$\{\boldsymbol{h}_{a_i}s_m\}$, which motivates the constellation-aware
anchor optimization introduced in Scheme~6.

\subsection{Joint constellation–distance anchoring (Scheme 6)}

Scheme~5 selects anchor ports by maximizing the minimum separation
between channel vectors, i.e., $\min_{i\neq j}\|\boldsymbol{h}_i-\boldsymbol{h}_j\|$.
However, the detection performance of IM is ultimately
determined by the Euclidean distance between the received signal
constellation points.

To better align the anchor design with the BER performance, we further
consider a joint constellation-distance criterion. Specifically, the
anchor set $\mathcal{E}$ is selected by maximizing the minimum distance
between all possible received signal points
\begin{equation}
D_{\min}^{\mathrm{joint}}(\mathcal{E})
=
\min_{(i,m)\neq(j,n)}
\big\|
\boldsymbol{h}_{a_i}s_m-\boldsymbol{h}_{a_j}s_n
\big\|.
\end{equation}

Compared with Scheme~5, which only considers the separation of channel
vectors, Scheme~6 directly optimizes the minimum Euclidean distance of
the joint spatial–symbol constellation, thereby providing a design
criterion more closely related to the BER performance.

The same greedy farthest-point selection framework used in
Scheme~5 can be directly applied by replacing the channel-distance
metric with the joint constellation distance defined above.
Therefore, the algorithmic procedure is omitted for brevity.

\subsection{Computational Complexity Discussion}

We briefly analyze the computational complexity of the
proposed representative-port selection strategies.

For Scheme~3, the representative port in each cell is
determined by selecting the candidate port with the largest
channel gain. Suppose that the $c$-th cell contains
$|\mathcal{S}_c|$ candidate ports. The search within that
cell requires $\mathcal{O}(|\mathcal{S}_c|)$ operations.
Summing over all cells yields the overall complexity $\mathcal{O}\!\left(\sum_{c=1}^{C}|\mathcal{S}_c|\right)
=
\mathcal{O}(Q),$
which scales linearly with the number of candidate ports.

For Scheme~4, the representative ports are obtained through
an iterative max--min optimization. In each iteration,
every candidate port in a cell evaluates its minimum channel
distance to the representative ports of the other cells.
Assuming that each cell contains approximately $Q/C$
candidate ports, the computational cost per iteration is
approximately $\mathcal{O}(CQ)$. If the algorithm converges
after $I$ iterations, the overall complexity becomes
$\mathcal{O}(ICQ)$.

For Scheme~5, the anchor set is optimized globally in the
channel domain using a farthest-point sampling strategy.
The computational cost is dominated by the pairwise channel
distance computation $\mathcal{O}(Q^2)$ and the iterative
anchor updates $\mathcal{O}(NQ)$. Therefore, the overall
complexity scales as $\mathcal{O}(Q^2 + NQ).$

Scheme~6 adopts the same greedy farthest-point framework as
Scheme~5 but replaces the channel-distance metric with the
joint constellation distance
$\|\boldsymbol{h}_{a_i}s_m-\boldsymbol{h}_{a_j}s_n\|$.
Since the algorithmic structure remains unchanged, Scheme~6
has the same asymptotic complexity $\mathcal{O}(Q^2 + NQ)$.
The only additional overhead arises from evaluating the
constellation-dependent distance metric, which introduces a
constant factor related to the modulation order $M$ but does
not change the complexity order.

In summary, Schemes~1--3 have relatively low computational
complexity since the representative ports are selected
independently within each cell. Scheme~4 introduces an
iterative optimization to enlarge the channel separation
among index states, resulting in higher computational
complexity but improved detection reliability. Schemes~5 and
6 remove the geometric partition constraint and perform
global anchor optimization in the channel domain, which
leads to the highest computational cost among the considered
strategies\footnote{Further complexity reduction and scalable anchor selection designs constitute an interesting direction for future work.}.

Nevertheless, the anchor-selection procedures in
Schemes~3--6 are executed only once during system
configuration and can be performed offline before data
transmission. Therefore, the additional computational
overhead does not significantly affect the real-time
operation of practical MA-IM systems.

To facilitate comparison and provide a clearer overview of the proposed designs, we summarize the key characteristics of all schemes in Table~\ref{tab:scheme_comparison}. The performance trends are provided as qualitative insights based on the design principles of each scheme, and will be quantitatively validated in the simulation section.

\begin{table*}[htp]
\centering
\caption{Comparison of the proposed MA-IM schemes in terms of design principle, complexity, and performance characteristics}
\begin{tabular}{c|c|c|c|c}
\hline
\textbf{Scheme} & \textbf{Anchor Design Principle} & \textbf{Search Scope} & \textbf{Complexity} & \textbf{Performance Trend (Qualitative)} \\
\hline
Scheme 1 
& Random activation within cell 
& Local (cell-wise) 
& High (receiver-side exhaustive) 
& Very low (baseline) \\

Scheme 2 
& Geometry-based center selection 
& Local (cell-wise) 
& Low 
& Moderate \\

Scheme 3 
& SNR-based representative selection 
& Local (cell-wise) 
& Low 
& Moderate \\

Scheme 4 
& Cell-constrained max--min channel separation 
& Local (cell-wise, optimized) 
& Medium 
& High \\

Scheme 5 
& Global max--min channel separation 
& Global (all sampling points) 
& High (offline) 
& Moderate--high \\

Scheme 6 
& Joint constellation-distance-aware optimization 
& Global (constellation-aware) 
& High (offline) 
& Highest \\

\hline
\end{tabular}
\label{tab:scheme_comparison}
\end{table*}

\section{Proposed detector}
\label{sec3}
In this section, ML detectors are developed for the five
schemes introduced in the previous section. Furthermore,
a low-complexity detection method is proposed for
Scheme~1 to reduce the computational burden of exhaustive
ML search.
\subsection{ML Detector}

For the MA-IM schemes introduced in Section~\ref{sec22}, the receiver
jointly detects the MA position index and the transmitted modulation
symbol using a maximum-likelihood (ML) detector. The ML detector
evaluates all feasible transmit hypotheses and selects the one that
minimizes the Euclidean distance between the received signal and the
corresponding channel response.

The general ML decision rule can be written as
\begin{equation}
\label{eq:ML_general}
(\hat{q},\hat{s})
=
\arg\min_{\substack{q\in\mathcal{Q}\\ s\in\mathcal{S}_M}}
\left\|
\boldsymbol{y}
-
E_s s\,\boldsymbol{h}_q
\right\|^2 ,
\end{equation}
where $\mathcal{Q}$ denotes the set of allowable MA sampling positions
and $\mathcal{S}_M$ represents the $M$-QAM constellation.

The main difference among the proposed schemes lies in the definition
of the candidate set $\mathcal{Q}$.

\textbf{Scheme~1:}
the MA position is randomly selected within the activated grid during
transmission. Since the receiver does not know the exact transmit
location, the ML detector must search over the entire sampling set
\begin{equation}
    \label{equ23}
    \mathcal{Q}_1=\{1,2,\ldots,Q\},
\end{equation}
resulting in a per-symbol detection complexity of
$\mathcal{O}(QM)$.

After obtaining the detected sampling-point index $\hat q_1$,
the corresponding grid index can be determined through the
mapping
\begin{equation}
\label{equ24}
\hat I_1 = \ell(\hat q_1),
\end{equation}
where $\ell(\cdot)$ maps a sampling-point index to its
associated grid label.

\textbf{Scheme~2--Scheme~5:}
the MA transmission is restricted to a set of $N$ representative anchor
positions. Let the anchor set be
\[
\mathcal{Q}_a=\{a_1,a_2,\ldots,a_N\}.
\]
Although the anchor-selection strategies differ across these schemes
(geometry-based, SNR-based, distance-based, and channel-domain
optimization), the resulting ML detector shares the same structure
\begin{equation}
\label{eq:ML_anchor}
(\hat{c},\hat{s})
=
\arg\min_{\substack{c\in\{1,\ldots,N\}\\ s\in\mathcal{S}_M}}
\left\|
\boldsymbol{y}
-
E_s s\,\boldsymbol{h}_{a_c}
\right\|^2 .
\end{equation}

Since the search is only performed over the $N$ anchor positions,
the per-symbol detection complexity is reduced to
$\mathcal{O}(NM)$.

Therefore, compared with Scheme~1, the representative-anchor-based
schemes significantly reduce the detection complexity while improving
the reliability of index detection through more structured anchor
design.

\subsection{Low-complexity two-stage detector for scheme 1}
As revealed by expressions \eqref{equ23}--\eqref{equ24}, Scheme~1 requires the ML detector to
evaluate the distance metric over all $Q$ sampling points, resulting in a
computational complexity that grows linearly with $Q$.  
Therefore, when $Q \gg N$, the complexity of Scheme~1 becomes significantly
higher than that of Schemes~2-5.  
More importantly, despite employing an exhaustive ML search, Scheme~1 does not
necessarily achieve better error performance.  
This is because the channel responses of adjacent sampling points are highly
similar, which leads to extremely small decision distances between candidate
positions.  
Consequently, the ML detector becomes more susceptible to noise-induced
confusion, thereby yielding a higher detection error rate.  
In this sense, Scheme~1 simultaneously suffers from the dual drawbacks of
the highest computational complexity and the weakest decision
separability.

\begin{algorithm}[t]
\caption{Two-stage detector with top-$K$ anchors}
\label{alg:two_stage}
\begin{algorithmic}[1]

\REQUIRE 
Received signal $\boldsymbol{y}$, 
anchor channels $\{\boldsymbol{h}_{a_c}\}_{c=1}^N$ from Scheme~1,
sampling-point channels $\{\boldsymbol{f}(\boldsymbol{t}_q)\}_{q=1}^Q$, 
QAM alphabet $\mathcal{S}_M$, 
grid partitions $\{\mathcal{S}_c\}_{c=1}^N$, 
and $K$ ($K<N$).

\ENSURE 
Detected grid index $\hat{I}_{\mathrm{two}}$ and symbol $\hat{s}_{\mathrm{two}}$.

\vspace{0.5em}
\STATE \textbf{ Stage~1: Coarse detection using anchor-based ML metrics}
\FOR{$c=1$ \TO $N$}
    \STATE Compute coarse ML metric
    \begin{equation}
        \label{eq:stage1_metric}
                D_c = \min_{s\in\mathcal{S}_M}
        \bigl\|\,\boldsymbol{y}-s\,\boldsymbol{h}_{a_c}\,\bigr\|^2 .
    \end{equation}
\ENDFOR
\STATE Sort $\{D_c\}$ in ascending order and select top-$K$ anchors:
\[
   \mathcal{C}_K = \{c_1, c_2, \ldots, c_K\}.
\]

\vspace{0.5em}
\STATE \textbf{ Stage~2: Refined search over the union of top-$K$ grids}
\STATE Construct refined candidate set
\[
   \mathcal{U}_K = \bigcup_{c\in\mathcal{C}_K} \mathcal{S}_c .
\]
\STATE Perform fine ML search
\begin{equation}
    \label{eq:stage2_union}
    (\hat{q}_{\mathrm{two}},\hat{s}_{\mathrm{two}})
=
\arg\min_{\substack{q\in\mathcal{U}_K \\ s\in\mathcal{S}_M}}
\bigl\|\,\boldsymbol{y}-s\,\boldsymbol{f}(\boldsymbol{t}_q)\,\bigr\|^2 .
\end{equation}

\vspace{0.5em}
\STATE \textbf{ Final output}
\STATE Map the detected sampling index to its grid
\[
   \hat{I}_{\mathrm{two}}
   = \ell(\hat{q}_{\mathrm{two}}).
\]

\STATE \textbf{return} $\big(\hat{I}_{\mathrm{two}}, \hat{s}_{\mathrm{two}}\big)$.

\end{algorithmic}
\end{algorithm}

To alleviate these issues, this subsection proposes a 
low-complexity two-stage detector with top-$K$ anchors, as shown in Algorithm \ref{alg:two_stage}.  
The key idea is as follows:  
(i)~first leverage the anchor points of Scheme~1 to perform a coarse search
over all $N$ anchors and select the top-$K$ grids with the smallest
anchor-based ML metrics;  
(ii)~then, within these $K$ high-confidence grids, refine the detection by
searching only among the sampling points belonging to the selected grids.

% The proposed two-stage detector proceeds as follows.

% \textbf{Stage~1 (coarse grid detection with top-$K$ anchors):}  
% For each anchor $a_c$, we compute the anchor-based ML metric
% \begin{equation}
%   D_c 
%   \triangleq 
%   \min_{s \in \mathcal{S}_M}
%   \big\| \boldsymbol{y} - s\,\boldsymbol{h}_{a_c} \big\|^2,
%   \quad c \in \{1,\ldots,N\}.
%   \label{eq:stage1_metric}
% \end{equation}
% Then we sort $\{D_c\}$ in ascending order and select the index set of the
% top-$K$ anchors as
% \begin{equation}
%   \mathcal{C}_K 
%   \triangleq 
%   \big\{ c_1, c_2, \ldots, c_K \big\},
% \end{equation}
% where $\{c_1,\ldots,c_K\} \subset \{1,\ldots,N\}$ correspond to the $K$
% smallest metrics in \eqref{eq:stage1_metric}.

% \textbf{Stage~2 (refined search over the union of top-$K$ grids):}  
% Conditioned on the coarse selection $\mathcal{C}_K$, the search is refined
% only within the sampling points associated with the top-$K$ grids:
% \begin{equation}
% \left( \hat{q}_{\mathrm{two}},\hat{s}_{\mathrm{two}} \right) =\mathrm{arg}\min_{q\in \mathcal{U} _K} \min_{s\in \mathcal{S} _M} \bigl\| \boldsymbol{y}-s\,\boldsymbol{f}(\boldsymbol{t}_q) \bigr\| ^2,
%   \label{eq:stage2_union}
% \end{equation}
% where the candidate set $\mathcal{U}_K$ is defined as the union
% \begin{equation}
%   \mathcal{U}_K 
%   \triangleq 
%   \bigcup_{c \in \mathcal{C}_K} \mathcal{S}_c .
% \end{equation}
% The final detected grid index is then given by
% \begin{equation}
%   \hat{I}_{\mathrm{two}} = \ell\big(\hat{q}_{\mathrm{two}}\big).
% \end{equation}
% where $\ell(\cdot)$ maps a sampling-point index to its corresponding grid label.

In terms of complexity, the full ML detector
requires $\mathcal{O}(QM)$ metric evaluations per symbol.  
By contrast, the proposed two-stage detector with top-$K$ anchors needs
$\mathcal{O}(NM)$ evaluations to compute $\{D_c\}$ in
\eqref{eq:stage1_metric}, plus $\mathcal{O}\big(|\mathcal{U}_K|\big)
\approx \mathcal{O}\big(MK\,Q/N\big)$
evaluations in the refined search of \eqref{eq:stage2_union} when the sampling
points are evenly distributed among grids.  
Hence, the overall complexity scales as $\mathcal{O}\big(MN + MK\,Q/N\big),$
which yields a substantial reduction compared to $\mathcal{O}(MQ)$ when
$Q \gg N$ and $K \ll N$.  
Moreover, since the coarse stage already prunes unlikely grids and
concentrates the fine search on high-confidence regions, the performance
loss relative to full ML is typically minor.

\section{Performance analysis}
\label{sec4}

\subsection{Throughput and Spectral-Efficiency Analysis}

We first analyze the achievable spectral efficiency (SE) of the
proposed MA-IM framework.
For a conventional MA system without IM, information is carried solely by the modulation
symbol $s\in\mathcal{S}_M$, where $|\mathcal{S}_M|=M$ denotes the
size of the $M$-QAM constellation. The resulting SE is
\begin{equation}
R_{\mathrm{MA}} = \log_2 M .
\end{equation}

In the proposed MA-IM framework, the movable region is
discretized into $N$ index states. During each channel use,
$\lfloor \log_2 N \rfloor$ bits are conveyed by the index of the
selected position, in addition to the $\log_2 M$ bits carried
by the modulation symbol. Hence, the overall throughput becomes
\begin{equation}
R_{\mathrm{MA\text{-}IM}}
=
\log_2 M + \left\lfloor \log_2 N \right\rfloor .
\end{equation}

Compared with the conventional MA system, the proposed MA-IM
scheme therefore achieves an SE gain of
\begin{equation}
\Delta R
=
\left\lfloor \log_2 N \right\rfloor .
\end{equation}

If all candidate sampling points are used as index states
($N=Q$), the maximal achievable throughput is
\begin{equation}
R_{\mathrm{MA\text{-}IM}}^{\max}
\approx
\log_2 M + \left\lfloor \log_2 Q \right\rfloor .
\end{equation}

This result indicates that the achievable SE of MA-IM grows
with the number of distinguishable spatial sampling points.
In practice, however, the number of usable index states is
limited by channel correlation and detection reliability,
which motivates the anchor-selection strategies proposed in
Section~\ref{sec22}.

\subsection{ABEP analysis}

This subsection provides a theoretical characterization of the
ABEP of the proposed MA-IM
framework under ML detection.

We focus on Schemes~2--6, in which each index state is mapped
to a unique and deterministic anchor position $a_c$.
This property greatly simplifies the analysis since the ML
detector only needs to examine the $N$ anchor states rather
than the entire spatial sampling grid.

For anchor $a_c$ and QAM symbol $s_m$, the noiseless received
signal can be written as
\begin{equation}
x_{c,m} = h_{a_c}s_m .
\end{equation}

Collecting all possible index–modulation pairs yields a joint
constellation
$\mathcal{X} = \{x_\ell\}_{\ell=1}^{\mathcal{L}},
\qquad
\mathcal{L}=NM .$

Under AWGN with noise variance $N_0$, the ML decision between
two distinct constellation points $x_\ell$ and $x_{\ell'}$
depends only on their Euclidean separation. The pairwise
error probability (PEP) is therefore
\begin{equation}
P(\ell \to \ell')
=
Q\!\left(
\frac{|x_\ell-x_{\ell'}|}{\sqrt{2N_0}}
\right),
\qquad
\ell'\neq\ell ,
\label{eq:pep}
\end{equation}
where $Q(\cdot)$ denotes the Gaussian $Q$-function.

Let $w(\ell,\ell')$ denote the Hamming distance between the
bit labels associated with $x_\ell$ and $x_{\ell'}$.
Following the classical union bound for digital modulation
\cite{proakis2001digital}, the conditional bit error
probability satisfies
\begin{equation}
P_b|x_\ell
\le
\frac{1}{B_{\mathrm{tot}}}
\sum_{\ell'\neq\ell}
w(\ell,\ell')P(\ell\to\ell'),
\end{equation}
where
$B_{\mathrm{tot}}
=
\log_2 M+\left\lfloor\log_2 N\right\rfloor$
is the total number of transmitted bits per channel use.

Assuming all constellation points are equiprobable, the ABEP
is upper-bounded by
\begin{equation}
\label{eq:abep_union_bound}
P_b
\le
\frac{1}{\mathcal{L}B_{\mathrm{tot}}}
\sum_{\ell=1}^{\mathcal{L}}
\sum_{\ell'\neq\ell}
w(\ell,\ell')
Q\!\left(
\frac{|x_\ell-x_{\ell'}|}{\sqrt{2N_0}}
\right).
\end{equation}

The bound in \eqref{eq:abep_union_bound} accounts for all
pairwise confusions in the joint index–modulation
constellation and therefore provides a tight approximation
of the BER performance in the high-SNR regime.

Importantly, \eqref{eq:abep_union_bound} shows that the error
performance of MA-IM is governed by the geometry of the joint
constellation $\{h_{a_c}s_m\}$. Consequently, anchor
selection plays a crucial role in shaping the Euclidean
distances among constellation points. In particular,
anchor-design strategies that enlarge the minimum joint
distance, such as the channel-domain max–min construction
(Scheme~5) and the joint constellation optimization
(Scheme~6), lead to improved ABEP performance.

\section{Simulation results}
\label{sec:simulation}
This subsection presents a comprehensive numerical evaluation of the
proposed MA-IM transmission framework. We examine the detection
performance of the three index-modulation schemes developed in this
paper.
All simulation results are obtained under a unified setup to ensure a
fair comparison. The movable antenna (MA) operates over a two-dimensional
region $\mathcal{C}=[-G,G]^2$ with sufficiently dense spatial sampling according to Lemma \ref{lem1}.
To reflect a realistic system configuration, we consider a
{carrier frequency of $1\,$GHz}, corresponding to a wavelength of
$\lambda = 0.3\,$m. Unless otherwise
specified, let $G$=1m \footnote{The multipath channel is generated according to a stochastic model, where the phases are randomly drawn following standard statistical distributions (e.g., Rayleigh fading with uniformly distributed phases). As such, no fixed multipath profile is specified. All simulation results are obtained by averaging over independent channel realizations to ensure statistical reliability.}. Since this work is the first to introduce the MA-IM concept, no existing benchmark 
 algorithms are available for direct comparison. Since this work is the first to introduce the MA-IM framework,
there are no existing benchmark algorithms specifically designed
for MA-IM.

Although FA-IM and FAG-IM schemes in \cite{zhu2024index,guo2025fluid} also
employ antenna-port indexing, their design philosophy differs
from the proposed MA-IM framework. In particular, the methods
in \cite{zhu2024index,guo2025fluid} assume a given set of antenna ports and focus on
designing IM transmission over these ports, whereas the present
work aims to select the optimal ports from a dense set of spatial
sampling points within the movable region. Moreover, the
approaches in \cite{zhu2024index,guo2025fluid} are developed for MIMO systems, while the
MA-IM framework considered in this paper focuses on the SISO case.
Therefore, a direct comparison is not straightforward\footnote{In fact, once the optimal ports are determined, the resulting
system reduces to a conventional port-index modulation structure.}. 
Instead, to demonstrate the SE advantage of
the proposed MA-IM framework, we compare it with conventional
QAM systems transmitting the same number of bits per channel use.

 %Nevertheless, the three schemes 
% proposed in this paper provide meaningful internal baselines through which their 
% relative performance can be evaluated.

\begin{figure}[htp]
	\centering
		{\includegraphics[width=0.40\textwidth]{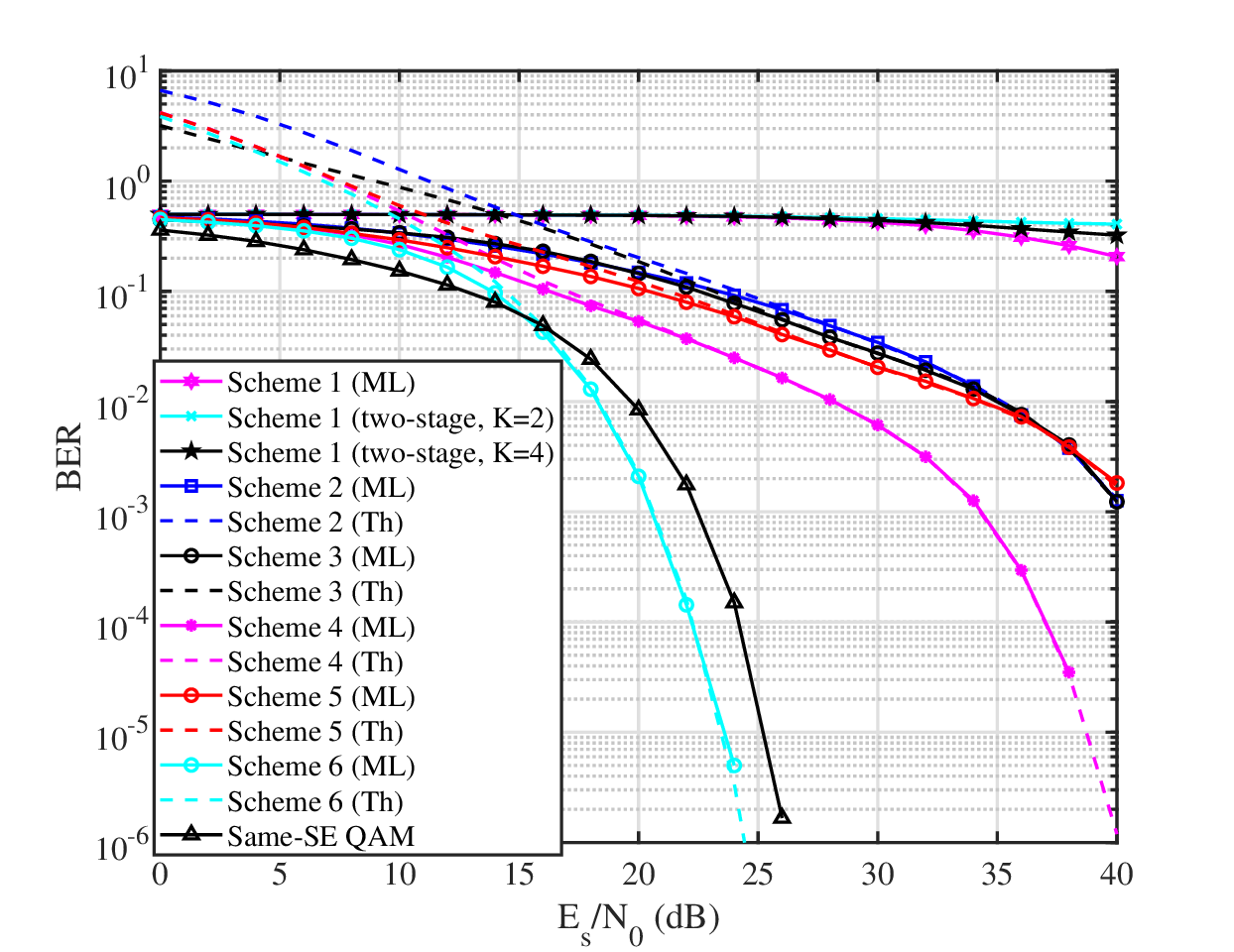}}	
		\caption{Simulation and analytical performance of the proposed MA-IM system when $L$=6, $N$=16, $\rho(\Delta r)= 0.9$ and 4QAM.}
		\label{fig6}	
\end{figure}

Fig.~\ref{fig6} reports the BER performance of the considered MA-IM schemes in a
six-path geometric multipath channel under QPSK modulation, together
with the corresponding analytical ABEP upper bounds. The spatial sampling
density is determined by the correlation constraint $\rho(\Delta r)=0.9$,
and the movable region is divided into $N=16$ index states. In the figure,
“ML’’ and “two-stage’’ denote Monte Carlo simulation results, while “Th’’
represents the analytical bound derived in~\eqref{eq:abep_union_bound}.
 The analytical bounds closely match the simulated BER in
the moderate-to-high $E_s/N_0$ regime, validating the accuracy of the
proposed analysis.

Several observations can be made. Scheme~1 exhibits poor performance due
to random in-cell activation, confirming its role as a baseline.
Schemes~2, 3, and 5 achieve comparable BER performance, indicating that
geometry-based anchoring, SNR-based selection, and unconstrained
channel-domain max--min design provide similar gains under the considered
setting.
In contrast, Scheme~4 achieves a noticeable improvement. This can be
explained by its cell-constrained max--min representative selection,
which jointly enforces spatial regularity and channel separation.
Compared with Scheme~5, which performs global max--min optimization
without geometric constraints, Scheme~4 avoids selecting overly clustered
or unevenly distributed anchors. As a result, it provides a more balanced
constellation structure and improves the effective separability among
index states.

Most importantly, Scheme~6 achieves a significant performance gain by
directly maximizing the minimum Euclidean distance of the joint
index–modulation constellation. As a result, it outperforms all other
schemes and approaches the same-SE QAM benchmark at high SNR. In addition,
the proposed two-stage detector achieves near-ML performance with much
lower complexity.
\begin{figure}[htp]
	\centering
		{\includegraphics[width=0.40\textwidth]{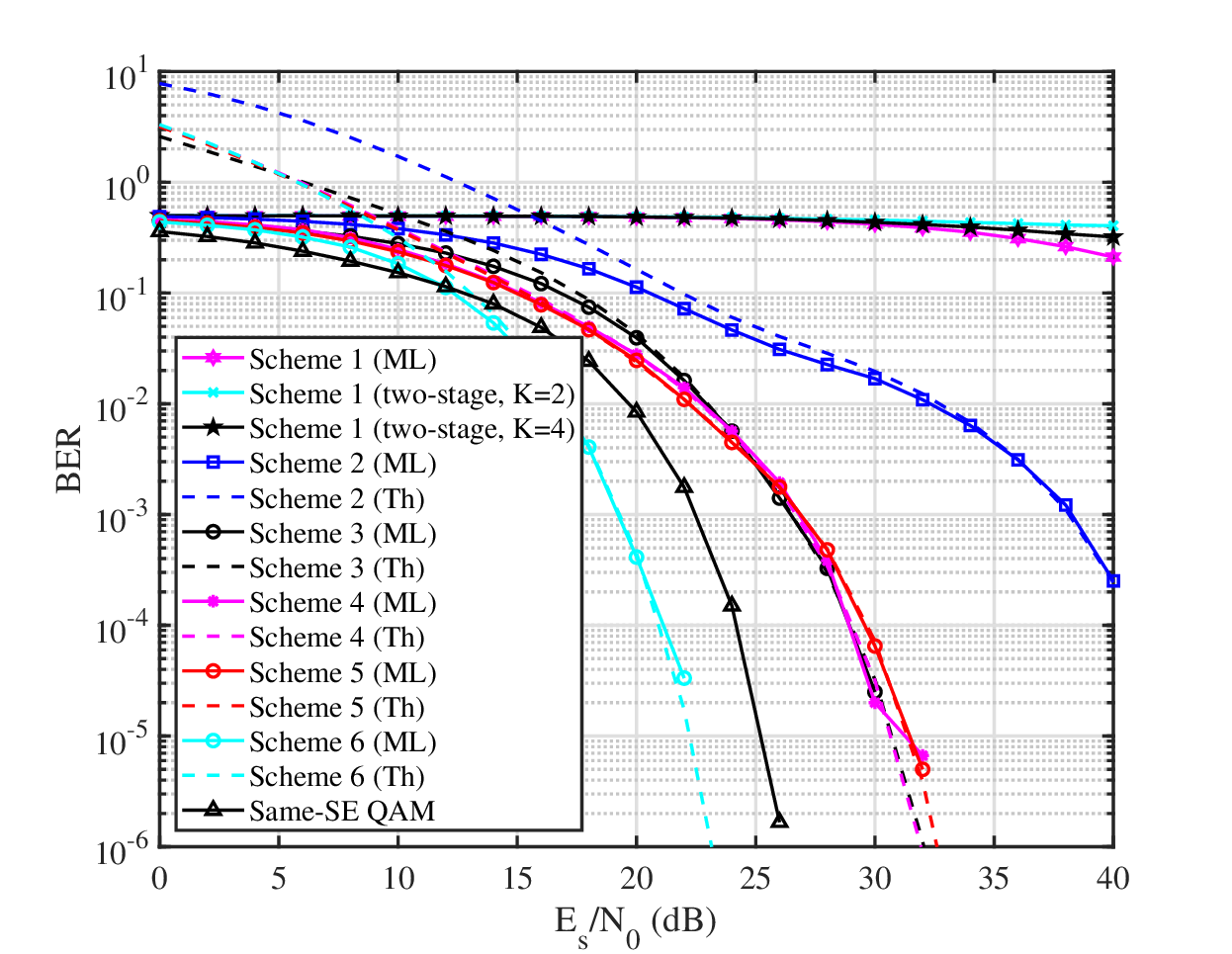}}	
		\caption{BER performance of the proposed MA-IM schemes  with \(L=12\), \(N=16\), $\rho(\Delta r)= 0.9$ and $M$=4.}
		\label{fig7}	
\end{figure}

When the number of propagation paths increases from $L=6$ in Fig.~\ref{fig6}
to $L=12$ in Fig.~\ref{fig7}, all MA-IM schemes exhibit a notable performance
improvement. This is due to the enhanced channel richness introduced by
additional multipath components, which reduces the similarity of channel responses across nearby positions
due to the superposition of multiple propagation paths with different
phases and directions.
As a result, the separability among index states is improved, leading to a
larger minimum distance in the joint index--modulation constellation and
thus lower BER. This gain is particularly pronounced for Scheme~3--Scheme~5,
which explicitly exploit channel-domain characteristics, while Scheme~2
achieves only moderate improvement due to its geometry-based design.

Scheme~6 continues to achieve the best performance, confirming that directly
optimizing the joint constellation distance is the most effective strategy,
especially in rich multipath environments. In contrast, the same-SE QAM
benchmark remains unchanged with $L$, highlighting that MA-IM uniquely
benefits from multipath by converting channel diversity into performance
gain. The proposed two-stage detector maintains near-ML performance across
all cases\footnote{Since Scheme~1 serves primarily as a baseline and consistently yields
poor performance, it is excluded from subsequent figures to better
highlight the relative gains among the proposed anchor optimization
strategies.}.

This result reveals that, unlike conventional modulation schemes,
MA-IM benefits from richer multipath environments by exploiting
the induced channel diversity as an additional information dimension.

\begin{figure}[htp]
	\centering
		\subfigure[$\rho(\Delta r)= 0.7$]{\includegraphics[width=0.4\textwidth]{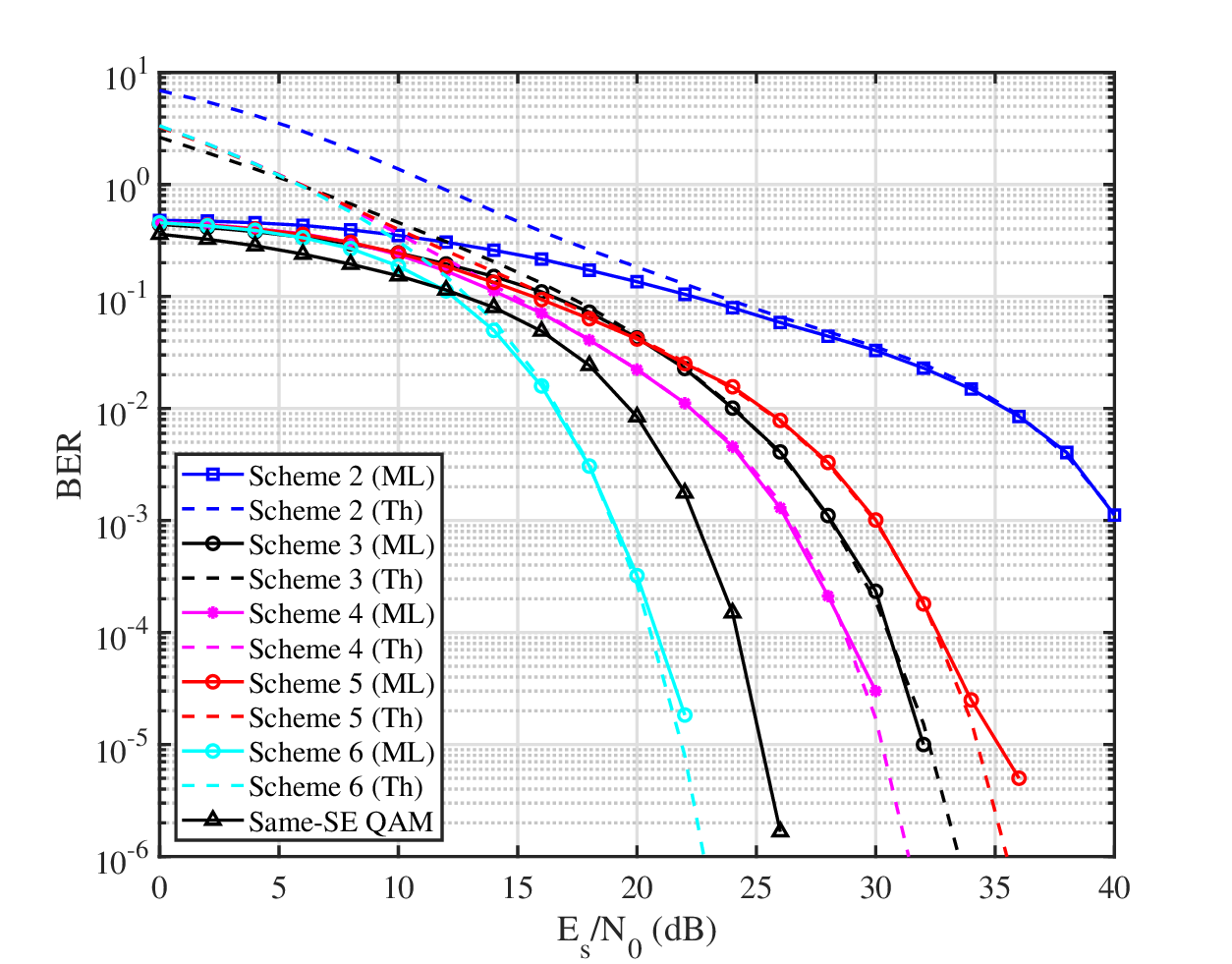}}	
        		\subfigure[$\rho(\Delta r)= 0.8$]{\includegraphics[width=0.4\textwidth]{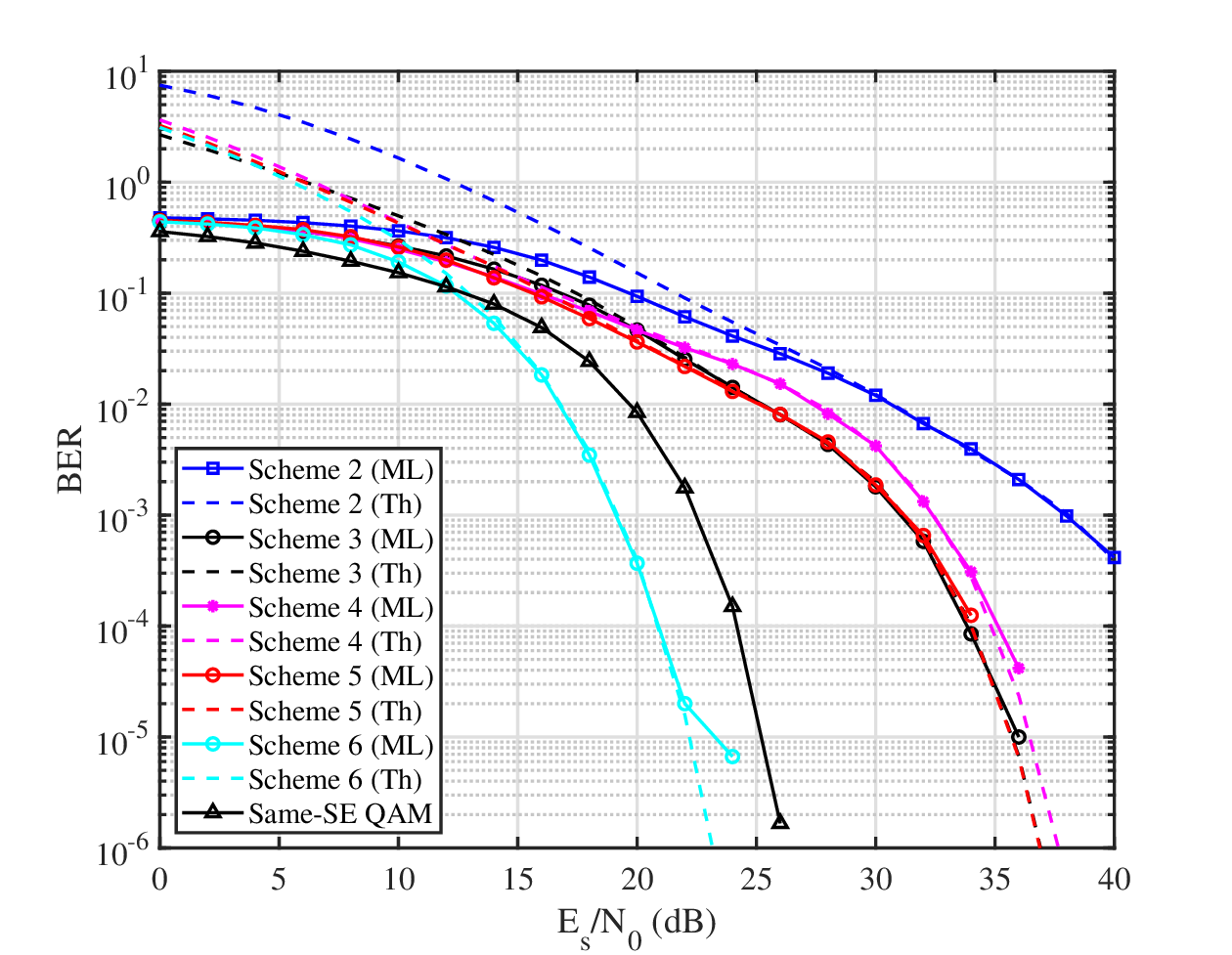}}	
		\caption{BER performance of the proposed MA-IM schemes  with \(L=12\), \(N=16\) and $M$=4.}
		\label{fig8}	
\end{figure}

Fig.~\ref{fig8} compares the BER performance of the proposed MA-IM
schemes under different spatial correlation levels, with
$\rho(\Delta r)=0.7$ and $\rho(\Delta r)=0.8$, respectively.
It can be observed that increasing $\rho(\Delta r)$ leads to a
noticeable performance degradation for all MA-IM schemes.
This is because a larger correlation coefficient implies higher
similarity among channel responses at different spatial positions,
which reduces the separability of index states and decreases the
minimum distance of the joint index--modulation constellation.

Despite this degradation, the relative performance ordering remains
consistent. Scheme~6 continues to achieve the best performance and
exhibits strong robustness against increased spatial correlation,
confirming the effectiveness of joint constellation optimization.
Schemes~3--5 experience more pronounced performance degradation as
$\rho(\Delta r)$ increases, since their designs explicitly rely on
channel-domain separability, which diminishes under stronger spatial
correlation. In particular, Scheme~4 is most sensitive due to its
dependence on max--min channel distance optimization.
In contrast, Scheme~2 exhibits more stable performance. As a
geometry-based scheme, it does not rely heavily on channel diversity,
and is therefore less affected by increased correlation. Consequently,
its relative performance improves compared to channel-aware schemes
in highly correlated environments.

In contrast, the same-SE QAM benchmark remains almost unchanged
under different $\rho(\Delta r)$ values, since it does not rely on
spatial degrees of freedom. This highlights a fundamental advantage
of MA-IM, namely its ability to leverage spatial diversity, while
also revealing that its performance is inherently coupled with the
spatial correlation structure of the channel. This observation
underscores the importance of anchor design under correlated
propagation conditions.

\begin{figure}[htp]
	\centering
				\includegraphics[width=0.4\textwidth]{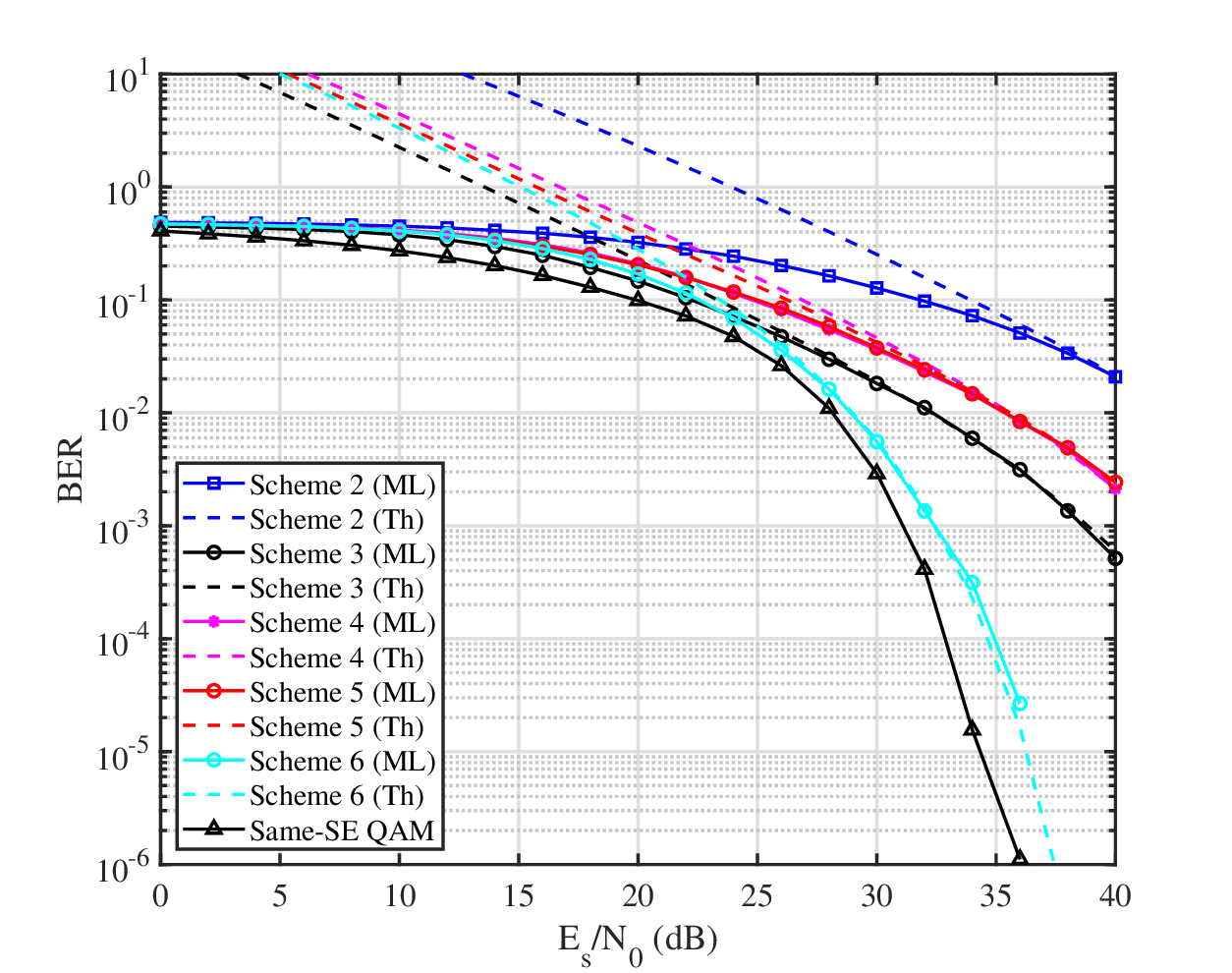}	
 				% \subfigure[$N$=64]{\includegraphics[width=0.48\textwidth]{results/BER_M4_N64_pho0.9_f1G.eps}}               
		\caption{BER performance of the proposed MA-IM schemes  with \(L=12\), $\rho(\Delta r)= 0.7$ and $M$=32, $N$=16.}
		\label{fig9}	
\end{figure}

Fig.~\ref{fig9} illustrates the BER performance of the proposed MA-IM
schemes under higher-order QAM modulation with $M=32$. Compared with
Fig.~\ref{fig8}(a), the BER performance of all schemes degrades as the
modulation order increases.
An interesting observation is that the same-SE QAM benchmark now
outperforms Scheme~6. This is because, although Scheme~6 directly
optimizes the Euclidean distances of the joint MA--QAM constellation,
its overall BER is still affected by both index-domain and symbol-domain
errors. In contrast, conventional QAM avoids index detection errors and
therefore becomes more competitive under high-order modulation.

Another notable observation is that Scheme~3 achieves significantly
better performance than Schemes~2, 4, and 5 in this setting. This
suggests that, when the modulation constellation becomes denser,
selecting anchors with stronger channel gains is more beneficial than
purely enlarging channel-domain separation. Meanwhile, the relatively
small gap between Schemes~4 and~5 indicates that max--min channel
distance alone is insufficient to guarantee the best BER performance
under high-order modulation.

Overall, these results reveal a modulation-order-dependent design
tradeoff: while joint constellation optimization is highly effective
under low-to-moderate modulation orders, SNR-oriented anchor selection
becomes increasingly attractive as the QAM order grows.

\begin{figure}[htp]
	\centering
\includegraphics[width=0.48\textwidth]{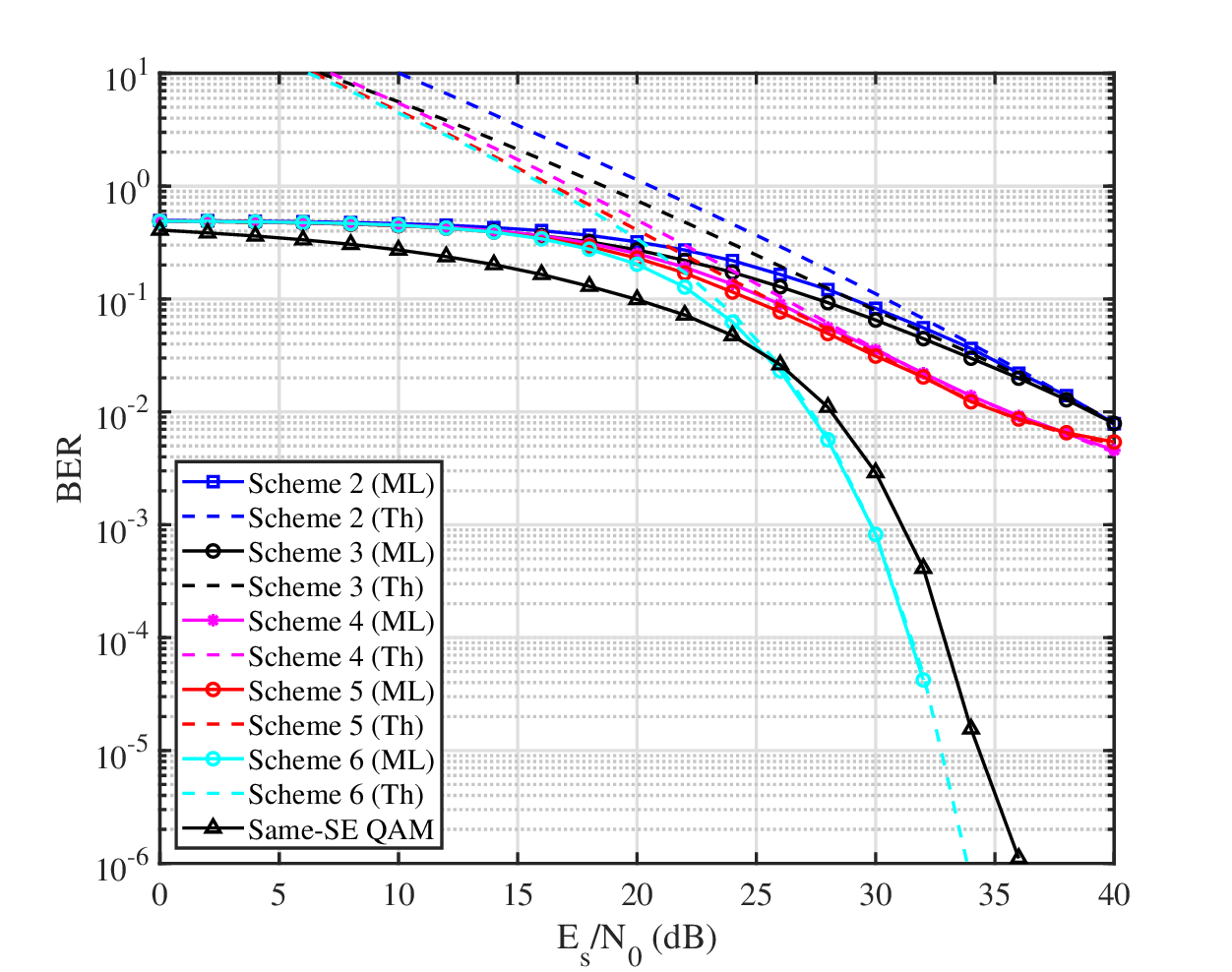}	              
		\caption{BER performance of the proposed MA-IM schemes  with \(L=12\), $\rho(\Delta r)= 0.7$ and $M$=4, $N$=128.}
		\label{fig10}	
\end{figure}

Fig.~\ref{fig10} shows the BER performance of the proposed MA-IM schemes
when the number of index states 128.
Compared with Fig.~\ref{fig8}(a), all schemes exhibit clear
performance degradation.
This behavior can be explained by the increased density of index states.
Although a larger $N$ improves the throughput, it also reduces the
minimum separation among candidate index states in the spatial/channel
domain, which in turn decreases the minimum distance of the resulting
joint constellation and degrades the BER performance. This further
confirms that not all sampling points should be directly indexed, and
that anchor optimization is necessary for reliable MA-IM transmission.

Another important observation from the comparison between
Fig.~\ref{fig9}--Fig.~\ref{fig10} is that increasing the QAM order and
increasing the number of MA index states affect MA-IM in fundamentally
different ways, even when the total spectral efficiency, i.e.,
$\log_2(M) + \log_2(N)$, is kept constant.
While a larger $M$ mainly densifies the symbol constellation, a larger
$N$ additionally reduces the separability among index states.
Therefore, enlarging $N$ imposes a more severe penalty on BER than
enlarging $M$.

Moreover, Scheme~6 remains the best-performing scheme even in this
dense-index regime, demonstrating the effectiveness of directly
optimizing the joint constellation distance. By contrast, Schemes~2
and~3 become the weakest anchor-based designs, indicating that simple
geometry-based or SNR-based selection is insufficient when the number of
index states is large. Schemes~4 and~5 achieve comparable performance,
suggesting that explicit distance-aware anchor design becomes more
important than purely local or gain-oriented rules as $N$ increases.

The simulation results demonstrate that the performance of the proposed
MA-IM framework is jointly governed by four key system parameters,
namely the number of propagation paths \(L\), the spatial sampling density
\(Q\), the number of index states \(N\), and the QAM modulation order \(M\).
Among them, \(N\) and \(M\) not only affect the BER performance but also
directly determine the SE of the system. This reveals
an inherent performance--rate tradeoff that must be carefully balanced
in practical implementations.

Among the considered schemes, Scheme~6 consistently achieves the best
BER performance across different configurations, owing to its direct
optimization of the minimum Euclidean distance of the joint
index--modulation constellation. Schemes~3--5 provide moderate performance
by partially exploiting channel-domain characteristics, whereas
Scheme~2, based on simple geometric partitioning, shows relatively limited
performance gains.

Nevertheless, it is worth noting that Scheme~2 remains closest to the
original movable antenna (MA) operation principle, as it preserves the
spatial continuity of antenna movement without enforcing a structured
anchor selection. This highlights a fundamental tradeoff between
implementation simplicity and performance optimization.

Overall, these results validate the effectiveness of the proposed
MA-IM framework and provide useful design guidelines for selecting
system parameters and anchor strategies in practical MA-enabled
wireless systems.

\section{Conclusion}
\label{sec6}

This paper investigated an MA-IM framework that exploits the spatial mobility of a single
reconfigurable antenna to create additional information-bearing
dimensions. By discretizing the continuous movable region into a dense
set of candidate sampling points and selecting representative anchors
for indexing, the proposed framework converts the spatial DoFs of the MA into a practical modulation resource.

Building upon this framework, we developed a family of anchor-selection strategies with different
levels of channel awareness, along with the corresponding ML and
low-complexity two-stage detectors, and derived unified ABEP upper
bounds based on the joint index--modulation constellation.

The results reveal several key insights. First, directly indexing all
sampling points is generally unreliable, highlighting the necessity of
anchor optimization. Second, MA-IM performance is governed by channel
richness, spatial correlation, the number of index states, and the
modulation order. In particular, increasing the number of propagation
paths improves performance via enhanced channel diversity, whereas
higher spatial correlation or excessively large index sets degrade BER
by reducing index-state separability. Moreover, increasing the QAM
order and the number of index states affect MA-IM differently, even
under the same transmission rate.

Among the considered schemes, the joint constellation-aware anchor
design consistently achieves the best performance, demonstrating that
optimizing channel-domain separation alone is insufficient and that
anchor design should align with the geometry of the resulting signal
constellation. With properly designed anchors, MA-IM can achieve
substantial performance gains and approach, or even outperform,
same-spectral-efficiency QAM baselines.
Overall, MA-IM provides a promising transmission paradigm whose key
advantage lies in optimizing representative ports from a dense movable
region before indexing.

Future work may extend the proposed framework to wideband channels,
MIMO systems, dynamic trajectory design, and prototype validation.
In addition, practical implementation aspects such as movement latency
and positioning errors of MA may affect system performance.
These effects can be mitigated through system-level designs, for example,
by employing multi-antenna architectures that enable parallel
transmission and repositioning.
A comprehensive investigation of such non-ideal effects, as well as
hardware constraints, remains an important direction for future work. 

\input{main_revised.bbl}

%\bibliographystyle{IEEEtran} 
%\bibliography{ref}

\end{document}